\documentclass{article}
\usepackage[margin=2cm]{geometry}

\usepackage{amssymb,amsmath,amsfonts,amsthm}
\usepackage{multicol,multirow}
\usepackage{hyperref,cleveref,enumerate,verbatim,lmodern}
\usepackage{listings,tikz}
\usepackage{sidecap,wrapfig}
\usetikzlibrary{calc}
\usetikzlibrary{positioning,automata,fit,arrows.meta}
\tikzset{arrows={[scale=1.1]}}
\tikzset{every edge/.style={draw,->,>=Latex,auto}}
\usepackage{algorithm,algorithmicx,calc}
\usepackage[noend]{algpseudocode}
\MakeRobust{\Call}
\makeatletter
\algnewcommand\algorithmicinput{\textbf{Input:}}\algnewcommand\Input{\item[\algorithmicinput]}
\algnewcommand\algorithmicoutput{\textbf{Output:}}\algnewcommand\Output{\item[\algorithmicoutput]}
\algnewcommand\Complexity{\item[\textbf{Complexity:}]}

\newcommand{\algorithmiccontinue}{\textbf{continue}}\newcommand{\Continue}{\State \algorithmiccontinue}
\newcommand{\Assert}[1]{\textbf{assert}(#1)}
\newcommand{\true}{\textbf{true}}
\newcommand{\false}{\textbf{false}}
\newcommand{\none}{\textbf{none}}

\makeatother

\usepackage{url}
\usepackage{xcolor}
\usepackage{anyfontsize}
\DeclareSymbolFont{rsfscript}{OMS}{rsfs}{m}{n}
\DeclareSymbolFontAlphabet{\mathrsfs}{rsfscript}

\newcommand{\Cerny}{\v{C}ern{\'y} }

\usetikzlibrary{decorations}
\pgfdeclaredecoration{dashsoliddouble}{initial}{
  \state{initial}[width=\pgfdecoratedinputsegmentlength]{
    \pgfmathsetlengthmacro\lw{.3pt+.5\pgflinewidth}
    \begin{pgfscope}
      \pgfpathmoveto{\pgfpoint{0pt}{\lw}}%
      \pgfpathlineto{\pgfpoint{\pgfdecoratedinputsegmentlength}{\lw}}%
      \pgfmathtruncatemacro\dashnum{%
        round((\pgfdecoratedinputsegmentlength-3pt)/6pt)
      }
      \pgfmathsetmacro\dashscale{%
        \pgfdecoratedinputsegmentlength/(\dashnum*6pt + 3pt)
      }
      \pgfmathsetlengthmacro\dashunit{3pt*\dashscale}
      \pgfsetdash{{\dashunit}{\dashunit}}{0pt}
      \pgfusepath{stroke}
      \pgfsetdash{}{0pt}
      \pgfpathmoveto{\pgfpoint{0pt}{-\lw}}%
      \pgfpathlineto{\pgfpoint{\pgfdecoratedinputsegmentlength}{-\lw}}%     
      \pgfusepath{stroke}
    \end{pgfscope}
  }
}
\newtheorem{theorem}{Theorem}
\newtheorem{corollary}[theorem]{Corollary}
\newtheorem{lemma}[theorem]{Lemma}
\newtheorem{proposition}[theorem]{Proposition}

\newtheorem{remark}[theorem]{Remark}
\numberwithin{theorem}{section}
\newtheorem{definition}{Definition}

\numberwithin{definition}{section}
\newtheorem{example}[theorem]{Example}

\renewcommand{\O}{\mathcal{O}}

\title{Completely reachable automata: a quadratic decision algorithm\\and a quadratic upper bound on the reaching threshold}
\author{Robert Ferens and Marek Szyku{\l}a\\
University of Wroc{\l}aw, Wroc{\l}aw, Poland\\
{\tt robert.ferens@cs.uni.wroc.pl}, {\tt msz@cs.uni.wroc.pl}}
\begin{document}
\maketitle
\begin{abstract}
A complete deterministic finite (semi)automaton (DFA) with a set of states $Q$ is \emph{completely reachable} if every nonempty subset of $Q$ is the image of the action of some word applied to $Q$.
The concept of completely reachable automata appeared, in particular, in connection with synchronizing automata; the class contains the \v{C}ern{\'y} automata and covers several distinguished subclasses.
The notion was introduced by Bondar and Volkov (2016), who also raised the question about the complexity of deciding if an automaton is completely reachable.
We develop an algorithm solving this problem, which works in ${\mathcal{O}(|\Sigma|\cdot n^2)}$ time and $\mathcal{O}(|\Sigma|\cdot n)$ space, where $n=|Q|$ is the number of states and $|\Sigma|$ is the size of the input alphabet.
In the second part, we prove a weak Don's conjecture for this class of automata: a nonempty subset of states $S \subseteq Q$ is reachable with a word of length at most $2n(n-|S|) - n \cdot H_{n-|S|}$, where $H_i$ is the $i$-th harmonic number.
This implies a quadratic upper bound in $n$ on the length of the shortest synchronizing words (reset threshold) for the class of completely reachable automata and generalizes earlier upper bounds derived for its subclasses.
\end{abstract}
%%%%%%%%%%%%%%%%%%%%%%%%%%%%%%%%%%%%%%%%%%%%%%%%%%%%%%%%%%%%
\section{Introduction}

We consider automata (DFAs) that are finite, deterministic, and have a completely defined transition function.
Furthermore, initial and final states are irrelevant to our problems.

The concept of completely reachable automata originates from the theory of synchronizing automata.
An automaton is \emph{synchronizing} if starting from all states, after reading a suitable word called a \emph{reset} word, we narrow the set of possible states where the automaton can be to a singleton.
On the other hand, an automaton is \emph{completely reachable} if starting from the set of all states, we can reach every nonempty subset of states.
In other words, every nonempty subset of states is the image of the action of some word applied to the set of all states.
Thus, every completely reachable automaton is synchronizing.

Synchronizing automata are most famous due to one of the most longstanding open problems in automata theory: the \Cerny conjecture from 1969, which states that every synchronizing $n$-state automaton admits a reset word of length at most $(n-1)^2$.
The currently known best upper bound on the length of the shortest reset words is cubic in $n$ \cite{Pin1983OnTwoCombinatorialProblems,Shitov2019,Szykula2018ImprovingTheUpperBound}: $\sim 0.1654\ n^3 + \O(n^2)$.
Bounds on this length, known as \emph{reset threshold}, were a topic of extensive studies for the general case and subclasses of automata, as well as the related computational problems.
Most of the research on the topic of synchronizing automata was collected and comprehensively described in the recent survey \cite{Volkov2022Survey}; we also refer to older ones \cite{KariVolkov2021Survey,Volkov2008Survey}.
Applications of synchronizing automata include, e.g., testing of reactive systems \cite{Sandberg2005Survey} and synchronization of codes \cite{BFRS21SynchronizingStronglyConnectedPartialDFAs,BPR2010CodesAndAutomata}.

\subsection{The decision problem}

The notion of completely reachable automata was first introduced in~2016 by Bondar and Volkov \cite{BV2016CompletelyReachableAutomata}, who also asked about the complexity of the computational problem of deciding whether a given automaton is completely reachable.
They showed that determining whether a given subset of states is reachable is PSPACE-complete \cite{BV2016CompletelyReachableAutomata}.

Later studies revealed a connection between completely reachable automata and the so-called Rystsov graphs, which, after a suitable generalization, can be used to characterize completely reachable automata \cite{BCV2023CompletelyReachableInterplay,BV2018CharacterizationOfCompletelyReachable}.
However, this does not yet lead to an effective algorithm, as it is unknown whether it is possible to compute these graphs in polynomial time.

The computational problem of whether an automaton is completely reachable got some partial solutions.
There was proposed an algorithm that in $\O(|\Sigma|\cdot n^2)$ time checks whether permutational letters together with singular letters of rank $n-1$, i.e., those whose action maps the set of all states to a subset of size exactly $n-1$, suffice for the complete reachability of a given automaton \cite{GJ2019HardlyReachableSubsets}.
Recently, the case of binary automata was solved with a quasilinear-time ($\O(n \log \log n)$) algorithm \cite{CV2022BinaryCompletelyReachable}, which strongly relies on the specificity of the actions of both letters if they ensure the complete reachability (e.g., one of the letters must act as a cycle on all states).

The analogous decision problem -- whether a given automaton is synchronizing -- is solvable in quadratic time in the number $n$ of states and linear in the size of the input alphabet $\Sigma$, i.e., in time $\O(|\Sigma|\cdot n^2)$ by a well-known algorithm \cite{Eppstein1990}.
The same applies to the problem of the reachability of a given (instead of any) singleton.
However, the question of whether there exists a faster algorithm remains a major open problem concerning synchronizing automata.

\subsection{Bounds on the reset threshold}

Concerning the bounds for the \Cerny problem (reset threshold), the class of completely reachable automata is particularly interesting, as it contains several previously studied subclasses.
First, it contains the \emph{\Cerny automata} \cite{Cerny1964}, which is the unique known infinite series of automata meeting the conjectured upper bound $(n-1)^2$.
It also includes some of the so-called \emph{slowly synchronizing} series \cite{AGV2013SlowlySynchronizing}, which achieve almost the same quantity (see~\cite{Maslennikova2019ResetComplexityOfIdeal} for a proof of their complete reachability).

Next, \emph{synchronizing automata with simple idempotents} are completely reachable (\cite{RS2023ResetThresholdsOfTransformationMonoids}).
These are automata such that each letter either acts as a permutation or its action maps exactly one state to another while fixing all the other states (the so-called \emph{simple idempotent}).
The obtained upper bound was quadratic ($2n^2-4n+2$) \cite{R2000EstimationSimpleIdempotents}, using an original method that introduced Rystsov graphs.
Recently, a reinvestigation of that method led to proving the bound $2n^2-6n+5$ for a slightly larger class of automata \cite{RS2023ResetThresholdsOfTransformationMonoids}.

The next subclass of completely reachable automata is called \emph{aperiodically 1-contracting automata} (\cite[Theorem~10]{D2016OneContracting}).
For the general case of aperiodically 1-contracting automata, there was no nontrivial upper bound found.
However, for the specific case where additionally there exists an \emph{efficient 1-contracting collection} of words, the \Cerny conjecture along with the stronger Don's conjecture has been proved (\cite[Theorem~2~(2)]{D2016OneContracting}).
This Don's conjecture states that for an $n$-state automaton and a subset $S$ of its states, if $S$ is reachable, then it can be reached with a word of length at most $n(n-|S|)$.
The conjecture was disproved in general \cite{GJ2019HardlyReachableSubsets} by constructing automata with subsets that are reachable but cannot be reached with a word shorter than $2^n/n$; these automata, however, have many unreachable subsets.
Thus, weaker variants of Don's conjecture were proposed: one restricting to completely reachable automata \cite[Problem~4]{GJ2019HardlyReachableSubsets}, and another, weaker but generally applicable to synchronizing automata, in relation to avoiding words \cite[Conjecture~15]{FSV21LowerBoundsOnAvoidingThresholds}.

Recently, progress has been made concerning Don's conjecture for binary completely reachable automata.
It has been partially confirmed in this class \cite{CV2023DonsConjectureForBinaryCompletelyReachable}, for the so-called \emph{standardized} binary completely reachable DFAs, yet disproved by exhibiting an infinite series of binary completely reachable automata with subsets of size $n-2$ that are not reachable with words shorter than $5/2 n - 3$ (whereas the Don's bound gives $2n$ in this case).
As for now, there are no other known nontrivial lower bounds for other cases than $|S|=n-2$.

The next studied subclass is the class of automata with a full transition monoid \cite{GGGJV2017BabaiAndCerny}.
These are automata where letters acting as generators of a transformation semigroup on the set of states yield all $n^n$ possible transformations, thus for every possible action there exists some word inducing it.
Clearly, complete reachability is a weaker property, as it requires the presence of just one transformation for each nonempty subset of states.
The upper bound obtained for automata with a full transition monoid was $2n^2-6n+5$; the technique behind it is a modification of Rystsov's method for simple idempotents.

Finally, it turned out that automata that contain a primitive permutation group together with a singular transformation of rank $n-1$ in its transition monoid are also completely reachable \cite{H2023NewCharacterizationsOfPrimitiveGroups}.
The synchronization of automata with a primitive group was particularly studied from the group theory perspective \cite{ACS2017SynchronizationAndItsFriends}, and for the mentioned case with a transformation of rank $n-1$ a question about their reset thresholds was stated explicitly \cite[Problem~12.36]{ACS2017SynchronizationAndItsFriends}.

However, for the whole class of completely reachable automata, only a cubic bound was known, although better than in the general case: $7/48 n^3 + \O(n^2)$ \cite{BCV2023CompletelyReachableInterplay}, which was obtained through the technique of avoiding words \cite{Szykula2018ImprovingTheUpperBound}, using the fact that all subsets of size $n-1$ must be reachable with short words.

Other studies where complete reachability appears include descriptional complexity of formal languages, where it is related or is a part of similar properties of automata such as sync-maximality \cite{H2021CompletelyReachablePrimitiveAndStateComplexity}, complete distinguishability \cite{H2023CompletelyDistinguishable}, and reset complexity \cite{Maslennikova2019ResetComplexityOfIdeal}.

\subsection{Contribution and paper organization}

We begin with basic properties in~\Cref{sec:Preliminaries}. 
Then, in~\Cref{sec:Witnesses}, we introduce the concept of a \emph{witness} set -- an unreachable set of states whose predecessor sets satisfy some additional properties.
Witnesses turn out to be an efficient tool for verifying and certifying the (non)complete reachability of an automaton.
We design a verification function and show that a set can be tested for being a witness in linear time.
In opposition to witnesses, as an auxiliary result, we show that verifying whether a subset has a larger predecessor set is PSPACE-complete; this is another variant in the family of preimage problems, not considered before \cite{BFS21PreimageProblems}.

In~\Cref{sec:PolynomialTimeAlgorithm}, we design a polynomial-time algorithm for the problem of deciding whether an automaton is completely reachable, which solves the open question from~2016~\cite{BV2016CompletelyReachableAutomata}.
It relies on finding a witness by building the so-called \emph{reduction graph}.
The algorithm with suitable optimizations works in $\O(|\Sigma|\cdot n^2)$ time and $\O(|\Sigma|\cdot n)$ space.

In~\Cref{sec:Bounds}, we prove that every nonempty subset $S$ of states in a completely reachable $n$-state automaton is reachable with a word of length at most $2n(n-|S|) - n \cdot H_{n-|S|}$, where $H_i$ is the $i$-th harmonic number.
This implies a weaker version (by the factor of $2$) of Don's conjecture stated for completely reachable automata \cite[Problem~4]{GJ2019HardlyReachableSubsets}.
The bound is considerably better than $2n$ for large sets $S$.
For this solution, we extend the function verifying a witness and develop an original complement-intersecting technique, which allows finding a short extending word for a given subset of states (i.e., that gives a larger preimage), provided that the subset and all larger subsets are reachable.

In~\Cref{sec:Conclusions}, we draw implications of the results.
We conclude that a completely reachable $n$-state automaton has a reset threshold of at most $2n^2 - n \ln n - 2n$ (for $n \ge 3$).
This does not prove the \Cerny conjecture for this class but its weakened variant by the factor of $2$.
Our bound is the first quadratic upper bound for the class of completely reachable automata (previously only a cubic one was known) and also improves the previously known quadratic upper bounds obtained with different techniques for the mentioned proper subclasses.
We discuss the impact of the size of group orbits in a completely reachable automaton on the upper bounds and observe that the \Cerny conjecture holds for completely reachable automata without permutational letters or, more generally, for automata where the maximum size of group orbits of the permutation group contained in the transition monoid is at most $\ln n$.
We also note that the complete reachability of only large sets is sufficient to derive a subcubic upper bound on the reset threshold.
Additionally, as a side corollary from our algorithm, we show that a completely reachable automaton has at most $2n-2$ letters that are necessary for keeping it completely reachable.

This paper is the extended version of a conference paper \cite{thisICALP} and contains several new results, in particular, a significantly improved algorithm working in quadratic time in $n$.

%%%%%%%%%%%%%%%%%%%%%%%%%%%%%%%%%%%%%%%%%%%%%%%%%%%%%%%%%%%%
\section{Preliminaries}\label{sec:Preliminaries}

A \emph{complete deterministic finite semiautomaton} (called simply \emph{automaton}) is a $3$-tuple $(Q,\Sigma,\delta)$, where $Q$ is a finite set of \emph{states}, $\Sigma$ is an \emph{input alphabet}, and $\delta\colon Q \times \Sigma \to Q$ is the \emph{transition function}, which is extended to a function $Q \times \Sigma^* \to Q$ in the usual way.
Throughout the paper, by $n$ we always denote the number of states in $Q$.
For a subset $S \subseteq Q$, by $\overline{S}$ we denote its complement $Q \setminus S$.

Given a~subset $S \subseteq Q$, the \emph{image} of $S$ under the action of a~word $w \in \Sigma^*$ is $\delta(S,w) = \{\delta(q,w) \mid q \in S\}$.
The \emph{preimage} of $S$ under the action of $w$ is $\delta^{-1}(S,w) = \{q \in Q \mid \delta(q,w) \in S\}$.
For a singleton $\{q\}$, we also simplify and write $\delta^{-1}(q,w)=\delta^{-1}(\{q\},w)$.
Note that $q \in \delta(Q,w)$ if and only if $\delta^{-1}(q,w) \ne \emptyset$.

The empty word is denoted by $\varepsilon$.
The \emph{rank} of a word $w \in \Sigma^*$ is $|\delta(Q,w)|$.
A word (or a~letter) $w$ is \emph{permutational} if it acts as a permutation on $Q$, thus if it has rank $n$.
Clearly, a permutational word consists of only permutational letters.
A nonpermutational word (or a~letter) is called \emph{singular}.

\subsection{Synchronization}

A \emph{reset word} is a word $w$ of rank $1$: $|\delta(Q,w)|=1$.
Equivalently, we have $\delta^{-1}(q,w) = Q$ for exactly one $q \in Q$.
If an automaton admits a reset word, then it is called \emph{synchronizing} and its \emph{reset threshold} is the length of the shortest reset words.

The central problem in the theory of synchronizing automata is the \emph{\v{C}ern\'{y} conjecture}, which states that every synchronizing $n$-state automaton has its reset threshold at most $(n-1)^2$.

\subsection{Reachability of sets}

For two subsets $S,T \subseteq Q$, if there exists a word $w \in \Sigma^*$ such that $\delta(T,w)=S$, then we say that $S$ is \emph{reachable from} $T$ \emph{with} the word $w$.
Then we also say that $T$ is a \emph{$w$-predecessor} of $S$.
It is simply a \emph{predecessor} of $S$ if it is a $w$-predecessor for some word $w$.
A set $S$ can have many $w$-predecessors, but if they exist, there is one maximal with respect to inclusion and size, which is the preimage $\delta^{-1}(S,w)$.

\begin{remark}\label{rem:predecessor_nonempty}
For $S \subseteq Q$ and $w \in \Sigma^*$, the preimage $\delta^{-1}(S,w)$ is a $w$-predecessor of $S$ if and only if $\delta^{-1}(q,w) \ne \emptyset$ for every state $q \in S$.
Equivalently, we have $\delta(\delta^{-1}(S,w))=S$.
\end{remark}

\begin{remark}\label{rem:predecessor_maximal}
For $S \subseteq Q$ and $w \in \Sigma^*$, if $\delta^{-1}(S,w)$ is a $w$-predecessor of $S$, then all $w$-predecessors of $S$ are contained in it, thus $\delta^{-1}(S,w)$ is maximal with respect to inclusion and size.
If $\delta^{-1}(S,w)$ is not a $w$-predecessor of $S$, then $S$ does not have any $w$-predecessors.
\end{remark}

A word $w$ is called \emph{extending} for a subset $S$, or we say that $w$ \emph{extends} $S$, if $|\delta^{-1}(S,w)| > |S|$.
It is called \emph{properly extending}\footnote{ In~\cite{CV2023DonsConjectureForBinaryCompletelyReachable}, a subset that admits a properly extending word is called \emph{expandable}.} if additionally $\delta^{-1}(S,w)$ is a $w$-predecessor of $S$.
This is equivalent to that $|\delta^{-1}(q,w)| \ge 1$ for all $q \in S$, and $|\delta^{-1}(q,w)| > 1$ for at least one $q \in S$.

A subset $S \subseteq Q$ is \emph{reachable} if $S$ is reachable from $Q$ with any word.
An automaton is \emph{completely reachable} if all nonempty subsets $S \subseteq Q$ are reachable.
Equivalently, $Q$ is a predecessor of all its nonempty subsets.
The latter leads to an alternative characterization of completely reachable automata:
\begin{remark}\label{rem:Characterization}
An automaton $(Q,\Sigma,\delta)$ is completely reachable if and only if for every nonempty proper subset of $Q$, there is a properly extending word.
\end{remark}

\textsc{Completely Reachable} is the following decision problem:
Given an automaton $(Q,\Sigma,\delta)$, is it completely reachable?

For a subset $S \subseteq Q$, the \emph{reaching threshold} is the length of the shortest words reaching $S$ from $Q$.
\emph{Don's conjecture} states that the reaching threshold of a nonempty subset $S$ is $\le n(n-|S|)$.
It is known to be false in general but holds in certain cases.

\begin{figure}[htb]\centering%
\begin{tikzpicture}[node distance=1.9cm,scale=1,every node/.style={transform shape}]
\tikzset{every loop/.style={min distance=.5cm,looseness=6}}
\node[state] (q0) {$q_0$};
\node[state] [right=of q0] (q1) {$q_1$};
\node[state] [below=of q1] (q2) {$q_2$};
\node[state] [left=of q2] (q3) {$q_3$};
\draw(q0) edge[dashed] node[auto,midway]{$a,b$} (q1);
\draw(q1) edge[dashed] node[auto,midway]{$a$} (q2);
\draw(q2) edge[dashed] node[auto,midway]{$a$} (q3);
\draw(q3) edge[dashed] node[auto,midway]{$a$} (q0);

\draw(q0) edge[out=-10,in=186,looseness=0.001] (q1);
\draw(q1) edge[loop,out=85,in=35,looseness=6] node[auto,midway]{$b$} (q1);
\draw(q2) edge[loop,out=-40,in=-80,looseness=6] node[auto,midway]{$b$} (q2);
\draw(q3) edge[loop,out=-100,in=-140,looseness=6] node[auto,midway]{$b$} (q3);
\end{tikzpicture}\hspace{1cm}%
\input{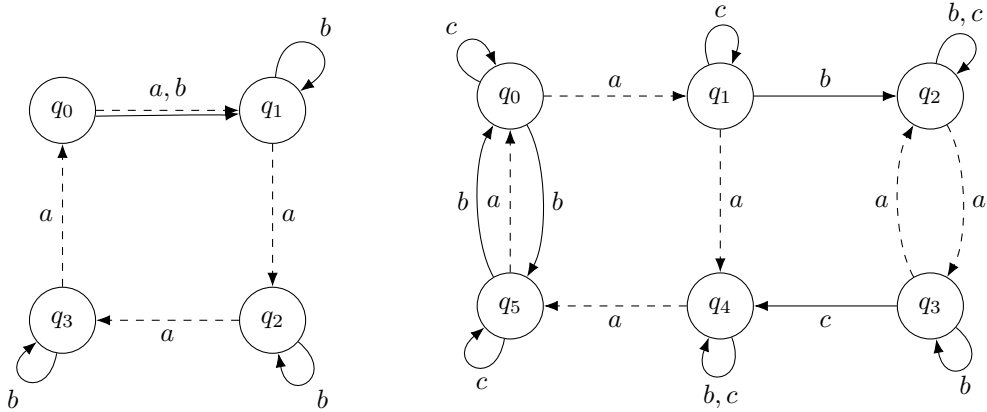}
\caption{Two completely reachable automata; the left one is from the \Cerny series \cite{Cerny1964}.
The transitions of permutational letters are dashed and those of singular letters are solid.}\label{fig:examples-completely}
\end{figure}

\begin{example}
\Cref{fig:examples-completely}~(left) shows the \Cerny automaton with $4$ states, which is completely reachable.
It could be demonstrated that the \Cerny automata meet the Don's upper bound for all $|S|$.
For $n=4$, the subsets with the largest reaching thresholds for each subset size are: $\{q_0,q_1,q_2\}$ $(4)$, $\{q_0,q_1\}$ $(8)$, and $\{q_0\}$ $(12)$.
\end{example}

\begin{example}\label{ex:Example1}
\Cref{fig:examples-completely}~(right) shows a completely reachable automaton that does not belong to any subclass of completely reachable automata mentioned in the introduction.
\end{example}
\begin{proof}[Proof of the example]
For complete reachability, we use the criterion from~\Cref{rem:Characterization}.
For every nonempty proper subset $S \subsetneq Q$, we find a properly extending word.

First, consider the subsets $S$ that contain exactly one from $\{q_1,q_2\}$.
If $q_1 \notin S$ and $q_2 \in S$, then letter $b$ is properly extending for $S$.
If $q_1 \in S$ and $q_2 \notin S$, then word $ca^3$ is properly extending for $S$.
Similarly, consider the subsets $S$ that contain exactly one from $\{q_3,q_4\}$.
If $q_3 \notin S$ and $q_4 \in S$, then letter $c$ is properly extending for $S$.
If $q_3 \in S$ and $q_4 \notin S$, then word $ba$ is properly extending for $S$.
Now, if $q_1,q_2 \in S$ and $q_5 \notin S$, or $q_1,q_2 \notin S$ and $q_5 \in S$, or $q_3,q_4 \in S$ and $q_0 \notin S$, or $q_3,q_4 \notin S$ and $q_0 \in S$, then by applying $a^{-2}$ (two times $a^{-1}$), we obtain a subset that fit in one of the above cases.

Since $S \notin \{\emptyset,Q\}$, there remain only two cases:
\begin{itemize}
\item $S = \{q_3,q_4,q_0\}$; then $\delta^{-1}(S,b) = \{q_3,q_4,q_5\}$, which fits in the previous case with $q_1 \notin S$ and $q_5 \in S$.
\item $S = \{q_1,q_2,q_5\}$, which is mapped back to the above by $a^{-1}$.
\end{itemize}

To see that the automaton does not belong to the mentioned classes, we make the following observations.

The automaton is ternary, thus not binary \cite{CV2023DonsConjectureForBinaryCompletelyReachable}.
Furthermore, removing any of the three letters yields a nonsynchronizing automaton, thus also not completely reachable.

The action of $b$ is not a permutation nor a simple idempotent, thus the automaton is not an \emph{automaton with simple idempotents} \cite{R2000EstimationSimpleIdempotents}.

The transition monoid is not full since \cite{GGGJV2017BabaiAndCerny}, e.g., the contained permutation group (generated solely by the action of $a$) is not transitive.
The contained permutation group is also not primitive \cite{RS2023ResetThresholdsOfTransformationMonoids}, since $\{q_0,q_1,q_4,q_5\}$ and $\{q_2,q_3\}$ are blocks preserved by the action of $a$.

For aperiodically 1-contracting subclass, recall from~\cite{D2016OneContracting} that a word of rank $n-1$ defines uniquely two states $q,q'$ such that $\delta^{-1}(q,w) = \emptyset$ and $|\delta^{-1}(q',w)|=2$; these states are commonly called \emph{excluded} and \emph{duplicated} states \cite{BCV2023CompletelyReachableInterplay}.
For such a word $w$, we assign the edge $(q \to q')$.
An automaton is \emph{aperiodically 1-contracting} if there exists a collection of $n$ words of rank $n-1$, whose edges form one full cycle on the states \cite[page~5]{D2016OneContracting}.
In our automaton, the complete list of edges of all words of rank $n-1$ is as follows:
\begin{align*}
(q_1 \to q_2)\;\text{by $b$},\quad(q_4 \to q_3)\;\text{by $ba$},\quad(q_5 \to q_2)\;\text{by $ba^2$},\quad(q_0 \to q_3)\;\text{by $ba^3$},\\
(q_3 \to q_4)\;\text{by $c$},\quad(q_2 \to q_5)\;\text{by $ca$},\quad(q_3 \to q_0)\;\text{by $ca^2$},\quad(q_2 \to q_1)\;\text{by $ca^3$}.
\end{align*}
These edges can be obtained with other words too, e.g., $(q_3 \to q_4)$ by $bac$.
Yet, the list of exhaustive because every such edge can be obtained with a word starting from a letter of rank $n-1$ followed by permutational letters.
We considered all words starting from $b$ or $c$ followed by $a$ repeated at most three times; more $a$ occurrences yield the same transformations.
From the list, we see that the edges form two separated components $\{q_1,q_2,q_5\}$ and $\{q_0,q_3,q_4\}$, so there is no collection of words whose edges form one full cycle.
\end{proof}

%%%%%%%%%%%%%%%%%%%%%%%%%%%%%%%%%%%%%%%%%%%%%%%%%%%%%%%%%%%%
\section{Witnesses}\label{sec:Witnesses}

To determine that an automaton $(Q,\Sigma,\delta)$ is not completely reachable, it is enough to find an unreachable nonempty subset $S \subsetneq Q$, which is then a counterexample to the complete reachability of the automaton.
However, the problem of determining whether a given set is (un)reachable is PSPACE-complete \cite{BV2016CompletelyReachableAutomata}.
Therefore, we must put additional restrictions on such sets that can witness the noncomplete reachability of the automaton so that verifying them is computationally easier.

First, consider nonempty subsets $S \subsetneq Q$ that do not admit a properly extending word, or equivalently, that do not have a larger predecessor.
Still, a set can have exponentially many predecessors of the same size.
In fact, the problem of determining this condition is also PSPACE-complete:
\begin{theorem}
The decision problem ``Given an automaton $(Q,\Sigma,\delta)$ and a subset $S \subseteq Q$, does there exist a properly extending word for $S$?''
is PSPACE-complete.
\end{theorem}
\begin{proof}
Several variants of the existence of an extending word have been considered \cite{BFS21PreimageProblems}, but the existing constructions do not fit for the case of a \emph{properly} extending word, hence it requires a new proof.
As this theorem has an auxiliary character for the other results, its proof is given in Appendix.
\end{proof}

The following observation allows inferring the existence of a larger predecessor indirectly.

\begin{lemma}\label{lem:UnionPredecessor}
Let $S, T \subseteq Q$ be distinct.
If a word $u \in \Sigma^*$ is properly extending for $S \cup T$, then $u$ is also properly extending for either $S$ or $T$.
\end{lemma}
\begin{proof}
Let $u$ be a properly extending word for $U = S \cup T \subsetneq Q$, thus there is a larger maximal $u$-predecessor $U'$ of $U$, i.e., $\delta(U',u) = U$, $\delta^{-1}(U, u) = U'$, and $|U'| > |U|$.
By~\Cref{rem:predecessor_nonempty}, we have $|\delta^{-1}(q,u)| \ge 1$ for all $q \in U$.
As this holds for all the states of $S$ and $T$, these sets also have their $u$-predecessors $\delta^{-1}(S,u)$ and $\delta^{-1}(T,u)$, respectively.
Suppose that $|\delta^{-1}(S,u)| = |S|$ and $|\delta^{-1}(T,u)| = |T|$.
Then $|\delta^{-1}(q,u)| = 1$ for all $q \in U$, which gives a contradiction with $|U'| > |U|$.
\end{proof}

It follows that if for a given set $S$, we find two distinct predecessors $T$ and $T'$ such that $T \cup T' \neq Q$, then the existence of a properly extending word for $T \cup T'$, implies that for either $T$ and $T'$, hence also for $S$.
As we show later, finding such a pair of predecessors or verifying its nonexistence can be done effectively.

\begin{remark}\label{rem:DisjointComplements}
For two sets $T, T' \subseteq Q$, the condition $T \cup T' = Q$ is equivalent to $\overline{T} \cap \overline{T'} = \emptyset$ (disjoint complements).
\end{remark}

We define our witness as follows\footnote{This differs from earlier definitions from~\cite{thisICALP}, where it is called a \emph{witness candidate}.}:

\begin{definition}[Witness]\label{def:witness}
A nonempty subset $S \subsetneq Q$ is a \emph{witness} if:
\begin{enumerate}
\item it does not have any larger predecessor, and
\item all its predecessors have pairwise disjoint complements (which are of the same size by~(1)).
\end{enumerate}
A \emph{maximal witness} is a witness of the largest size among all witnesses of the automaton.
\end{definition}

Maximal witnesses are of special interest since they are precisely all the unreachable sets of maximal size and they can be effectively found exhaustively -- our algorithm solving the problem of complete reachability finds a maximal witness and can find all of them.
Let us summarize the properties of a witness.

\begin{corollary}\label{cor:Witness}
The following hold:
\begin{enumerate}
\item A witness is unreachable.
\item All predecessors of a witness are also witnesses.
\item Each unreachable nonempty subset $S \subsetneq Q$ of the largest size is a maximal witness.
\item An automaton is not completely reachable if and only if there exists a witness.
\end{enumerate}
\end{corollary}
\begin{proof}
(1) is immediate by~\Cref{def:witness}~(1).

For~(2), if $T$ is a predecessor of a subset $S$, then all predecessors of $T$ are also predecessors of $S$.
Hence, if $S$ is a witness, then all predecessors of $T$ also satisfy both conditions from~\Cref{def:witness}.

For~(3), let $S \subsetneq Q$ be a nonempty subset of the largest size among all unreachable subsets.
Then $S$ does not admit a properly extending word.
If we take any two distinct predecessors $T$ and $T'$ of $S$, then $T \cup T'$ is larger than $S$.
If they violate \Cref{def:witness}~(2), then $T \cup T \neq Q$, so $T \cup T$ admits a properly extending word, because $S$ is a largest unreachable set.
Then by~\Cref{lem:UnionPredecessor}, this word is also properly extending for $T$ or $T$', so in either case, also $S$ admits a properly extending word, which is a contradiction.
Thus, all pairs of distinct predecessors of $S$ satisfy \Cref{def:witness}~(2), thus $S$ is a witness, and there are no larger witnesses, as they would be larger unreachable sets.

For~(4), a witness is unreachable and nonempty, so if it exists, the automaton is not completely reachable.
Conversely, if there is any unreachable set, then there exists an unreachable set of the largest size, which is a witness by~(3).
\end{proof}

However, there exist automata with unreachable nonempty sets that are not witnesses and with witnesses that are not maximal.
\Cref{fig:examples-notcompletely} shows two nontrivial examples of not completely reachable automata.

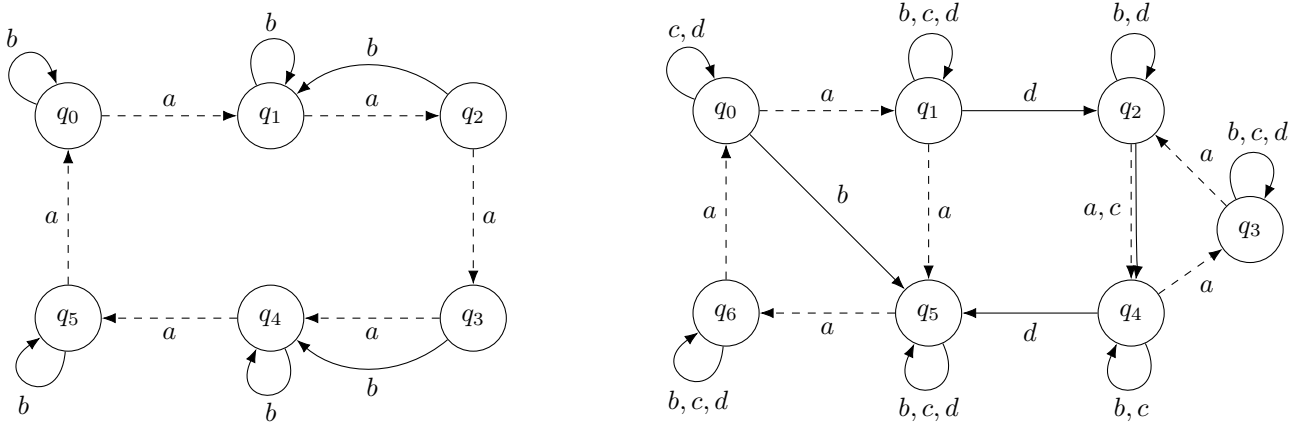
\begin{figure}[htb]\centering
\begin{tikzpicture}[node distance=1.8cm,scale=1,every node/.style={transform shape}]
%\tikzset{every edge/.style={draw,->,>=stealth,auto,semithick}}
\tikzset{every loop/.style={min distance=.5cm,looseness=6}}
\node[state] (q0) {$q_0$};
\node[state] [right=of q0] (q1) {$q_1$};
\node[state] [right=of q1] (q2) {$q_2$};
\node[state] [below=of q2] (q3) {$q_3$};
\node[state] [left=of q3] (q4) {$q_4$};
\node[state] [left=of q4] (q5) {$q_5$};
\draw(q0) edge[dashed] node[auto,midway]{$a$} (q1);
\draw(q1) edge[dashed] node[auto,midway]{$a$} (q2);
\draw(q2) edge[dashed] node[auto,midway]{$a$} (q3);
\draw(q3) edge[dashed] node[auto,midway]{$a$} (q4);
\draw(q4) edge[dashed] node[auto,midway]{$a$} (q5);
\draw(q5) edge[dashed] node[auto,midway]{$a$} (q0);
\draw(q0) edge[loop,out=160,in=110,looseness=6] node[auto,pos=.66]{$b$} (q0);
\draw(q1) edge[loop,out=115,in=65,looseness=6] node[auto,midway]{$b$} (q1);
\draw(q2) edge[loop,out=140,in=40,looseness=1] node[above,midway]{$b$} (q1);
\draw(q3) edge[loop,out=-140,in=-40,looseness=1] node[below,midway]{$b$} (q4);
\draw(q4) edge[loop,out=-65,in=-115,looseness=6] node[auto,midway]{$b$} (q4);
\draw(q5) edge[loop,out=-95,in=-145,looseness=6] node[auto,pos=.4]{$b$} (q5);
\end{tikzpicture}%\hspace{.01cm}%
\hfill\begin{tikzpicture}[node distance=1.8cm,scale=1,every node/.style={transform shape}]
%\tikzset{every edge/.style={draw,->,>=stealth,auto,semithick}}
\tikzset{every loop/.style={min distance=.5cm,looseness=6}}
\node[state] (q0) {$q_0$};
\node[state] [right=of q0] (q1) {$q_1$};
\node[state] [right=of q1] (q2) {$q_2$};
\node[state] [below right=.95cm and .95cm of q2] (q3) {$q_3$};
\node[state] [below=of q2] (q4) {$q_4$};
\node[state] [left=of q4] (q5) {$q_5$};
\node[state] [left=of q5] (q6) {$q_6$};

\draw(q0) edge[dashed] node[auto,midway]{$a$} (q1);
\draw(q1) edge[dashed] node[auto,midway]{$a$} (q5);
\draw(q5) edge[dashed] node[auto,midway]{$a$} (q6);
\draw(q6) edge[dashed] node[auto,midway]{$a$} (q0);
\draw(q2) edge[dashed,out=-90,in=90] node[auto,midway]{} (q4);
\draw(q2) edge[out=-82,in=82,looseness=0.01] node[auto,midway,swap,xshift=-.08cm]{$a,c$} (q4);
\draw(q4) edge[dashed,swap] node[auto,midway]{$a$} (q3);
\draw(q3) edge[dashed,swap] node[auto,midway]{$a$} (q2);

\draw(q0) edge node[auto,midway]{$b$} (q5);

\draw(q4) edge node[auto,midway]{$d$} (q5);
\draw(q1) edge node[auto,midway]{$d$} (q2);

\draw(q0) edge[loop,out=160,in=110,looseness=6] node[auto,pos=.68]{$c,d$} (q0);
\draw(q1) edge[loop,out=115,in=65,looseness=6] node[auto,midway]{$b,c,d$} (q1);
\draw(q2) edge[loop,out=115,in=65,looseness=6] node[auto,midway]{$b,d$} (q2);
%\draw(q3) edge[loop,out=25,in=-25,looseness=6] node[auto,midway]{$b,c,d$} (q3);
\draw(q3) edge[loop,out=115,in=65,looseness=6] node[auto,midway,xshift=.1cm]{$b,c,d$} (q3);
\draw(q4) edge[loop,out=-65,in=-115,looseness=6] node[auto,midway]{$b,c$} (q4);
\draw(q5) edge[loop,out=-65,in=-115,looseness=6] node[auto,midway]{$b,c,d$} (q5);
\draw(q6) edge[loop,out=-95,in=-145,looseness=6] node[auto,pos=.38]{$b,c,d$} (q6);
%\useasboundingbox (q6) rectangle (q3);
\end{tikzpicture} 
\caption{Two automata that are not completely reachable. The transitions of permutational letters are dashed and those of singular letters are solid.
}\label{fig:examples-notcompletely}
\end{figure}

\begin{example}
The automaton from~\Cref{fig:examples-notcompletely}~(left) is not completely reachable and:
\begin{itemize}
\item All sets of size $5$ are maximal witnesses.
\item The sets (of size $4$) $Q \setminus \{q_i,q_{(i+1) \bmod 6}\}$ for $i \in \{0,\ldots,5\}$ are reachable.
\item The sets (of size $4$) $Q \setminus \{q_i,q_{(i+2) \bmod 6}\}$ for $i \in \{0,\ldots,5\}$ are not properly extensible (only $a^{-1}$ yields a predecessor), but they are not witnesses, since their complements are not pairwise disjoint.
\item The set $\{q_0,q_2,q_4\}$ and its complement $\{q_1,q_3,q_5\}$ are witnesses but not maximal.
\item All singletons are reachable, e.g, $\delta(Q,(ba)^3 b) = \{q_1\}$; the automaton is synchronizing.
\end{itemize}
\end{example}

\begin{example}
The automaton from~\Cref{fig:examples-notcompletely}~(right) is not completely reachable and:
\begin{itemize}
\item All sets of size $6$ are reachable.
\item Maximal witnesses have size $5$ and they are precisely $Q \setminus \{q_0,q_5\}$ and $Q \setminus \{q_1,q_6\}$.
The other sets of size $5$ are reachable.
\item Sets $\{q_0,q_1,q_5,q_6\}$ and $\{q_2,q_3,q_4\}$ are witnesses and their only predecessors are themselves, respectively.
\item All sets of size $2$ and $1$ are reachable; the automaton is synchronizing.
\end{itemize}
\end{example}

\subsection{Verifying a witness}

The following observation will be useful for efficient testing of whether a given set is a witness:
\begin{lemma}\label{lem:WitnessPredecessorCount}
A witness $S \subsetneq Q$ has at most $\lfloor n/(n-|S|)\rfloor$ predecessors (including itself), where $n = |Q|$.
\end{lemma}
\begin{proof}
The complements of the predecessors of $S$ are pairwise disjoint, so each state is contained in at most one complement.
Each of these complements has the same size equal to $n-|S|$, so there are at most $\lfloor n/(n-|S|)\rfloor$ predecessors.
\end{proof}

We build a polynomial procedure that tests whether a given set $S$ is a witness; it is shown in~\Cref{alg:IsWitness}.
Starting from $S$, we process all its predecessors in a breath-first search manner; for this, a FIFO queue $\mathit{Process}$ is used.
The next set $T$ is taken from the queue in line~6.
Then, we verify whether $\overline{T} \cap \overline{T'} \neq \emptyset$, for some previously processed set $T'$, which is the condition from~\Cref{rem:DisjointComplements}.
For this, we maintain $\mathit{Absent}$ array, which for every state $q$ indicates whether $q$ occurred in the complement of some previously processed set, and if so, this set is stored as $\mathit{Absent}[q]$.
We additionally use this array to check whether the same set $T$ has been processed previously (line~9); if so, then it is ignored.
Otherwise, when $T \ne T'$, we know that $S$ is not a witness as the complements of its two predecessors have a nonempty intersection.
We update $\mathit{Absent}$ array in lines~12--13.
Finally, in lines~14--19, we compute and add one-letter predecessors of $T$ to the queue.
If one of these predecessors is larger than $T$, then this immediately implies that $S$ is not a witness.
Since predecessors are never smaller than the set, all processed sets that are put into the queue are of the same size $|S|$.
When all predecessors of $S$ have been considered and neither of the two conditions occurred, the function reports that $S$ is a witness.

\begin{algorithm}[htb!]
\caption{Verifying whether a given subset is a witness from \Cref{def:witness}.}\label{alg:IsWitness}
\begin{algorithmic}[1]
\Input An $n$-state automaton $\mathrsfs{A}=(Q,\Sigma,\delta)$ and a nonempty $S \subsetneq Q$.
\Output \textbf{true} if $S$ is a witness; \textbf{false} otherwise.
\Complexity $\O(|\Sigma|\cdot n^2)$ time (\Cref{lem:IsWitness}); $\O(|\Sigma|\cdot n)$ time if optimized (\Cref{lem:IsWitnessOptimized}).
\Function{IsWitness}{$\mathrsfs{A},S$}
\State $\mathit{Process} \gets \Call{EmptyFifoQueue}{ }$ \Comment{It contains predecessors of $S$ to be processed}
\State $\mathit{Process}.\Call{Push}{S}$
\State $\mathit{Absent} \gets$ Array indexed by $q \in Q$ initialized with \textbf{none}
\While{\textbf{not} $\mathit{Process}.\Call{Empty}{ }$}
  \State $T \gets \mathit{Process}.\Call{Pop}{ }$
  \If{$\mathit{Absent}[q] \neq \none$ for some $q \in \overline{T}$} \Comment{Then $T \cup T' \neq Q$}
    \State $T' \gets \mathit{Absent}[q]$
    \If{$T = T'$}
      \Continue \Comment{$T$ has been processed previously}
    \EndIf
    \State \Return \false \Comment{$T$ and $T'$ are predecessors with nondisjoint complements}
  \EndIf
  \ForAll{$q \in \overline{T}$}
    \State $\mathit{Absent}[q] \gets T$
  \EndFor
  \ForAll{$a \in \Sigma$}
    \If{$T$ has $a$-predecessor}
      \State $T' = \delta^{-1}(T,a)$
      \If{$|T'| > |S|$}
        \State \Return \false
      \EndIf
      \State $\mathit{Process}$.\Call{Push}{$T'$}
    \EndIf
  \EndFor
\EndWhile
\State \Return \true \Comment{All predecessors of $S$ were checked}
\EndFunction
\end{algorithmic}
\end{algorithm}

Here we do a simplified analysis of the time complexity of~\Call{IsWitness}{}.
Then, in the next subsection, we employ a technical optimization that reduces its running time to linear.

\begin{lemma}\label{lem:IsWitness}
Function \Call{IsWitness}{} from~\Cref{alg:IsWitness} is correct and can be implemented to work in $\O(|\Sigma|\cdot n^2)$ time.
\end{lemma}
\begin{proof}
\noindent\emph{Correctness:}
The function can end in line~11, 18, or~20.
The first case (line~11) means that $S$ has two distinct predecessors, $T$ and $T'$ such that $\overline{T} \cap \overline{T'} \neq \emptyset$, which implies that $S$ is not a witness.
The second case (line~18) means that we have found a larger predecessor of $S$, thus $S$ is not a witness.
The last possibility (line~20), where the function ends with a positive answer, occurs when there are no more predecessors of $S$ to consider ($\mathit{Process}$ becomes empty), so all predecessors of $S$ have been checked and the two previous cases have not occurred.
Thus, $S$ satisfies the conditions from~\Cref{def:witness}.

\noindent\emph{Running time:}
We separately consider the time complexity of two types of iterations of the main loop: \emph{full} and \emph{extra} iterations.
An \emph{extra} iteration is where the case from line~7 holds; then the iteration ends either in line~10 or~11.
Otherwise, we count it as a \emph{full} iteration.

In each iteration, $T$ has size $|S|$, because a predecessor cannot be smaller than the set and we check in line~17 if we encounter a larger one.
The number of distinct processed predecessors (including $S$ itself) of the same size $|S|<n$ with pairwise disjoint complements is at most $\lfloor n / (n-|S|)\rfloor$ by~\Cref{lem:WitnessPredecessorCount}.
Hence, this bounds the number of full iterations.
An extra iteration is either the one ending the algorithm or one with an already processed set $T$.
As these repeated sets must have been added previously in full iterations, and one full iteration adds at most $|\Sigma|$ sets to the queue, the number of extra iterations is at most $|\Sigma|\cdot \lfloor n / (n-|S|)\rfloor$.

A full iteration can be performed in $\O(|\Sigma|\cdot n)$ time if just in $\mathit{Absent}$ we store a pointer/reference to $T$ (line~13 in constant time) instead of copying the subset.
An extra iteration can be trivially performed in $\O(n)$ time.

Summing the total cost of full and extra iterations, we get the upper bound:
\[ \lfloor n / (n-|S|)\rfloor \cdot \O(|\Sigma|\cdot n) + |\Sigma|\cdot \lfloor n / (n-|S|)\rfloor \cdot \O(n) = \O(|\Sigma|\cdot n^2) .\]
\end{proof}

Verifying a witness in polynomial time already leads to an algorithm in co-NP solving the main problem.

\begin{corollary}
Problem \textsc{Completely Reachable} can be solved in co-NP.
\end{corollary}
\begin{proof}
To certify that a given automaton $\mathrsfs{A}$ is not completely reachable, we guess a witness $S \subsetneq Q$ and call $\Call{IsWitness}{\mathrsfs{A},S}$ to verify it.
If the automaton is not completely reachable, then there exists some witness.
Otherwise, there are no witnesses, thus the function returns {\false} regardless of the chosen $S$.
\end{proof}

\subsection{Optimization of the representation of sets}

We employ a technical optimization of altering the representation of sets.
When applied to the function verifying a witness, we achieve linear time, which is optimal.
The described optimization is also crucial for our polynomial algorithm (\Cref{sec:PolynomialTimeAlgorithm}) solving the main problem.
It also decreases the running time by $\O(n)$; without it, the running time would remain cubic in $n$.

We replace the usual array-based representation (the characteristic vector) of a set with a list of states in its complement.
This reduces the time complexity because the algorithms spend most of the time processing large sets, and their sizes are correlated with the maximum number of iterations.
If a subset $T$ is in the form of a list of states in its complement, we can do an operation on $T$ in $\O(n-|T|)$ time instead of $\O(n)$.
Yet, this requires some preprocessing to compute predecessors effectively, i.e., computing $\delta^{-1}(T,a)$ and checking if it is a predecessor of $T$.
This preprocessing must be done once at the beginning of the algorithm.

\begin{lemma}\label{lem:ComplementList}
For an automaton $(Q,\Sigma,\delta)$, after a preprocessing performed in $\O(|\Sigma|\cdot n)$ time and space, if a set $T$ is stored in the form of a list of states in its complement $\overline{T}$, then for a letter $a \in \Sigma$, we can check if $\delta^{-1}(T,a)$ is $a$-predecessor and, if it is, store this preimage in same form in $\O(n-|T|)$ time.
\end{lemma}
\begin{proof}
The preprocessing is as follows.
For each letter $a \in \Sigma$ and each state $q \in Q$, we compute the list $P_{q,a}$ containing the states from $\delta^{-1}(q,a)$.
Note that for the same letter $a$, these lists are pairwise disjoint and sum up to $Q$ over $q \in Q$.
Additionally, for each $a \in \Sigma$, we also compute its rank $|\delta(Q,a)|$ stored as $\mathit{rank}(a)$.
Note that $n-\mathit{rank}(a)$ is the number of states $q$ such that $\delta^{-1}(q,a)| = 0$.
This preprocessing is easily doable in $\O(|\Sigma|\cdot n)$ time and space.

Now, suppose that we are given to check and compute the preimage $\delta^{-1}(T,a)$, where $T$ is given in the form of a list of states in $\overline{T}$.
First, we count the states $q \in \overline{T}$ with an empty preimage, i.e., such that $|\delta^{-1}(q,a)| = 0$; for this, we look whether the preprocessed $P_{q,a}$ is empty, thus this takes $\O(n-|T|)$ time.
The preimage is $a$-predecessor of $T$ if and only if the number of the states with an empty preimage equals $n-\mathit{rank}(a)$, because if it is smaller, then some other state $p \in T$ has an empty preimage $\delta^{-1}(p,a)$.
Second, if the test passes, we build a list of states by concatenating all the lists $P_{q,a}$ for $q \in \overline{T}$, and the resulting list represents $\delta^{-1}(\overline{T},a) = \overline{\delta^{-1}(T,a)}$.
Since a predecessor is never smaller than the set, hence its complement is never larger than the complement of the set, the length of our list is at most $|\overline{T}| = n-|T|$, so we can do this computation in $\O(n-|T|)$ time.
\end{proof}

\begin{lemma}\label{lem:IsWitnessOptimized}
Function~\Call{IsWitness}{} from~\Cref{alg:IsWitness} can be implemented to work in $\O(|\Sigma|\cdot n)$ time and space.
\end{lemma}
\begin{proof}
\noindent\emph{Running time}:
At the beginning of the function, we run the processing from~\Cref{lem:ComplementList}.
The initialization in lines~2--4 is performed in $\O(n)$ time.

We store the sets in $\mathit{Process}$ and also $T$ and $T'$ in the form of a list of states in their complements.
Recall from the proof of~\Cref{lem:IsWitness} that the number of full iterations (those where the condition from line~7 does not hold) is $\le \lfloor n / (n-|S|)\rfloor$ and the number of extra iterations (the condition from line~7 holds) is $\le |\Sigma|\cdot\lfloor n/(n-|S|)\rfloor$.
Lines~5--13 can be performed in $\O(n-|S|)=\O(n-|T|)$ time if $\mathit{Absent}$ stores references/pointers to set representations; note that in lines~7 and~12 we iterate over $q \in \overline{T}$, and each access to $\mathit{Absent}[q]$ is done in constant time.
Using the optimization from~\Cref{lem:ComplementList}, also lines~15 and~16 can be performed in $\O(n-|S|)$ time.
Lines~17--19 are trivially performed in constant time.
Therefore, a full iteration takes $\O(|\Sigma|\cdot (n-|S|))$ time and an extra iterations takes $\O(n-|S|)$ time.
Altogether, we get the time complexity:
\[ \O(|\Sigma|\cdot n) + \lfloor n / (n-|S|)\rfloor \cdot \O(|\Sigma|\cdot (n-|S|)) + |\Sigma|\cdot \lfloor n / (n-|S|)\rfloor \cdot \O(n-|S|) = \O(|\Sigma|\cdot n). \]

\noindent\emph{Running space}: We note that apart from the automaton and the preprocessed structures, which require $\O(|\Sigma|\cdot n)$ space, the space used by the function fits in $\O(n)$ space because the sum of the lengths of all the lists stored in $\mathit{Process}$ does not exceed $n$.
\end{proof}

The running time of~\Call{IsWitness}{} is now optimal, because, in the worst case, any algorithm must read all transitions of the automaton:

\begin{proposition}\label{pro:IsWitnessOptimal}
$\Omega(|\Sigma|\cdot n)$ is a lower bound on the worst-case time of verifying whether a given set is a witness in a given automaton.
\end{proposition}
\begin{proof}
We construct an adversarial family showing that any correct algorithm must check all the transitions of the given automaton.

First, consider an automaton $\mathrsfs{A}$ with the set of states $Q = \{q_0,q_1,\ldots,q_{n-1}\}$ and over the alphabet $\Sigma = \{a_0,a_1,\ldots,a_{|\Sigma|-1}\}$; we assume $|\Sigma| \ge 2$.
Let each letter $a_j \in \Sigma$ act as the same cyclic permutation: $\delta(q_i,a_j) = q_{(i+1) \bmod n}$ for all $q_i$.
Let $S = Q \setminus \{q_0\}$, which is a witness in $\mathrsfs{A}$.
Now, for each $0 \le i \le n-1$ and $0 \le j \le |\Sigma|-1$, we define $\mathrsfs{A}_{i,j}$ to be the automaton $\mathrsfs{A}$ with one transition modified: $\delta(q_i,a_j) = q_i$.
Then in each $\mathrsfs{A}_{i,j}$, $S$ can be properly extended to $Q$, since $\delta(Q,a_j) = Q \setminus \{q_{(i+n-1) \bmod n}\}$ and we can reach $S$ by the cyclic permutation of any another letter.
Hence, automaton $\mathrsfs{A}$ together with automata $\mathrsfs{A}_{i,j}$ form an adversarial family, and every algorithm verifying whether $S$ is a witness, in the worst case, must check all the transitions to distinguish $\mathrsfs{A}$ from the others. 
\end{proof}

\subsection{Reduction for witness containment}

To build a polynomial algorithm to solve the main problem, we will need to find a witness.
In this task, we will rely on \emph{reduction for witness containment}.
Suppose that we are given a subset $S \subsetneq Q$ and want to find a witness that is contained in $S$.
If we know a properly extending word $w$ for $S$, then we can reduce $S$ to its proper subset by removing some states which surely cannot be in any witness contained in $S$.

\begin{lemma}\label{lem:ReductionByPredecessor}
Let $S \subsetneq Q$ and $\delta^{-1}(S,w)$ be a $w$-predecessor of $S$ for some word $w$.
Then for every witness $S' \subseteq S$, every state $q \in S'$ is such that $|\delta^{-1}(q,w)| = 1$.
\end{lemma}
\begin{proof}
Since $\delta^{-1}(S,w)$ is a $w$-predecessor of $S$, by~\Cref{rem:predecessor_nonempty}, we have $\delta^{-1}(q,w) \ne \emptyset$ for every state $q \in S$.
Suppose for a contradiction that there exists a witness $S' \subseteq S$ that contains a state $p \in S'$ such that $|\delta^{-1}(p,w)| > 1$.
Note that the set $\delta^{-1}(S',w)$ is a $w$-predecessor of $S'$ since its superset $S$ has a $w$-predecessor, and we have
\[ |\delta^{-1}(S',w)| = |\delta^{-1}(p,w)| + \sum_{s \in S \setminus \{p\}} |\delta^{-1}(s,w)| > 1+(|S|-1) = |S|. \]
Thus, $S'$ cannot be a witness, since it has a larger predecessor.
\end{proof}

It follows that to reduce $S$ for witness containment, we need to find a properly extending word for $S$.
Then we can remove at least one state from $S$ so that all the witnesses in $S$ are still contained in the resulting smaller subset.

%%%%%%%%%%%%%%%%%%%%%%%%%%%%%%%%%%%%%%%%%%%%%%%%%%%%%%%%%%%%%%%%%%%%%%%%%%%%%%%%%%%%%%%%%%%%%%%%%%%%%%%
%%%%%%%%%%%%%%%%%%%%%%%%%%%%%%%%%%%%%%%%%%%%%%%%%%%%%%%%%%%%%%%%%%%%%%%%%%%%%%%%%%%%%%%%%%%%%%%%%%%%%%%
%%%%%%%%%%%%%%%%%%%%%%%%%%%%%%%%%%%%%%%%%%%%%%%%%%%%%%%%%%%%%%%%%%%%%%%%%%%%%%%%%%%%%%%%%%%%%%%%%%%%%%%
\section{A polynomial-time algorithm}\label{sec:PolynomialTimeAlgorithm}

\newcommand{\vertexOf}{\Call{vertexOf}{}_\mathcal{R}}
\newcommand{\nextVertex}{\Call{nextVertex}{}_\mathcal{R}}
\newcommand{\scopeOf}{\Call{scopeOf}{}_\mathcal{R}}
\newcommand{\setOf}{\Call{setOf}{}_\mathcal{R}}

We describe the overall idea of the algorithm.
The next subsections provide a detailed technical description and proofs.

To search for a witness, we maintain a structure storing information about the possible existence of witnesses.
This structure is a \emph{reduction graph} and we denote it by $\mathcal{R}$.
It allows inferring reductions in the following way: if a given subset $S$ does not contain a state $q$, then we can find a proper subset $S' \subsetneq S$ such that all witnesses contained in $S$ are also contained in $S'$.
Thus, $S$ is reduced to $S'$ for witness containment.
The reduction graph allows the effective use of implications of all reductions computed so far.

The vertices of the reduction graph are pairwise disjoint subsets of $Q$ that altogether sum up to $Q$.
Hence, for each state $q \in Q$, there is a unique vertex with the set containing $q$; we denote it by $\vertexOf(q)$.

The existence of a vertex $V \subseteq Q$ denotes that if a subset $S$ does not contain a state from $V$ (one missing state is sufficient), then we can reduce $S$ for witness containment by removing all states in $V$.
An edge $V_1 \to V_2$ means that if a subset $S$ does not contain a state from $V_1$, then we can reduce $S$ for witness containment by removing all states in $V_2$, thus all in $V_1 \cup V_2$.
This relation is transitive but may be asymmetric, hence the edges are directed.

The reduction graph is always acyclic and every vertex has the out-degree either $0$ or $1$.
The vertices without an outgoing edge are \emph{leaves}, and the other have exactly one outgoing edge.

In the beginning, the reduction graph has $n$ disconnected singleton vertices.
Whenever we compute a new reduction, we add a suitable edge to the reduction graph; the outgoing edges are added only to leaves.
If this produces a cycle, then we replace the cycle with one vertex which is the union of the vertices in the cycle.
In this way, we keep the property that the reduction graph is a forest with the maximum out-degree $1$.
This makes it small (linear), avoiding the need for a quadratic number of edges, and trivial to traverse.

From the reduction graph, we can also extract the next subset to reduce.
The leaves of the smallest size are called \emph{minimal leaves}.
We choose a minimal leaf $V$ and consider its complement set $S = \overline{V}$.
We search for a reduction for $S$ by applying the inverse action of all letters and looking at the resulting preimages.
Consider a preimage $T = \delta^{-1}(S,a)$ for a latter $a \in \Sigma$ such that $T$ is $a$-predecessor of $S$.
There are three cases:
\begin{enumerate}
\item If $|T| > |S|$, then we have found a one-letter properly extending word.
Hence, by~\Cref{lem:ReductionByPredecessor}, we can reduce $S$ for witness containment by a state $r \in S$ such that $|\delta^{-1}(r,a)|>1$, which is also effectively computed (in linear time) as the word has only one letter.
As $r \notin \overline{S}$, it belongs to a different vertex than $V$ in the reduction graph.
Hence we can add a new edge to the reduction graph, which goes from $V$ to that vertex containing $q$.
\item If $|T| = |S|$ and $\overline{T}$ is not a minimal leaf, then we can find a reduction indirectly using the reduction graph.
We do this by choosing any state $q \in \overline{T}$ and following the path from $\vertexOf(q)$ up to a vertex containing any state $p \in T$.
Since $p$ is in a vertex reachable from $\vertexOf(q)$ and $q \in \overline{T}$, based on the reduction graph, we can reduce $T$ for witness containment by removing $p$.

Now, we need to remap the reduction of $T$ by $p$ to a reduction of $S$.
We know that we could apply $a^{-1}$ to every witness contained in $S$, and every predecessor of a witness is also a witness (\Cref{cor:Witness}~(2)).
It follows that if there exists a witness $W \subseteq S$ that contains the state $\delta(p,a)$, then also its predecessor $\delta^{-1}(W,a)$ must be a witness which contains $p$.
However, we already know from the reduction graph that every witness contained in $T$, thus in particular $\delta^{-1}(W,a) \subseteq T$, does not contain $p$.
Hence, we conclude that every witness contained in $S$ does not contain $\delta(p,a)$.
We add a new edge from $V$ to $\vertexOf(\delta(p,a))$.

Furthermore, in this case, it is important that the chosen vertex $V$ is a \emph{minimal} leaf.
Then $\overline{T}$ has the same size and is not a minimal leaf, which guarantees that from any vertex, we will eventually find a state $p \in T$.
\item If $\overline{T}$ is a minimal leaf, then we skip it and continue with the other letters.
\end{enumerate}
The third case may hold for all considered one-letter predecessor preimages $T$ of all $S = \overline{V}$, where $V$ is a minimal leaf.
Then this means that all these subsets $S$ are witnesses.

\begin{figure}[htb!]
\centering\begin{tikzpicture}[node distance=1.8cm,scale=1,every node/.style={transform shape}]
%\tikzset{every edge/.style={draw,->,>=stealth,auto,semithick}}
\tikzset{every loop/.style={min distance=.5cm,looseness=6}}
\node[state] (q0) {$q_0$};
\node[state] [right=of q0] (q1) {$q_1$};
\node[state] [right=of q1] (q2) {$q_2$};
\node[state] [below right=.95cm and .95cm of q2] (q3) {$q_3$};
\node[state] [below=of q2] (q4) {$q_4$};
\node[state] [left=of q4] (q5) {$q_5$};
\node[state] [left=of q5] (q6) {$q_6$};

\draw(q0) edge[dashed] node[auto,midway]{$a$} (q1);
\draw(q1) edge[dashed] node[auto,midway]{$a$} (q5);
\draw(q5) edge[dashed] node[auto,midway]{$a$} (q6);
\draw(q6) edge[dashed] node[auto,midway]{$a$} (q0);
\draw(q2) edge[dashed,out=-90,in=90] node[auto,midway]{} (q4);
\draw(q2) edge[out=-82,in=82,looseness=0.01] node[auto,midway,swap,xshift=-.08cm]{$a,c$} (q4);
\draw(q4) edge[dashed,swap] node[auto,midway]{$a$} (q3);
\draw(q3) edge[dashed,swap] node[auto,midway]{$a$} (q2);

\draw(q0) edge node[auto,midway]{$b$} (q5);

\draw(q4) edge node[auto,midway]{$d$} (q5);
\draw(q1) edge node[auto,midway]{$d$} (q2);

\draw(q0) edge[loop,out=160,in=110,looseness=6] node[auto,pos=.68]{$c,d$} (q0);
\draw(q1) edge[loop,out=115,in=65,looseness=6] node[auto,midway]{$b,c,d$} (q1);
\draw(q2) edge[loop,out=115,in=65,looseness=6] node[auto,midway]{$b,d$} (q2);
%\draw(q3) edge[loop,out=25,in=-25,looseness=6] node[auto,midway]{$b,c,d$} (q3);
\draw(q3) edge[loop,out=115,in=65,looseness=6] node[auto,midway,xshift=.1cm]{$b,c,d$} (q3);
\draw(q4) edge[loop,out=-65,in=-115,looseness=6] node[auto,midway]{$b,c$} (q4);
\draw(q5) edge[loop,out=-65,in=-115,looseness=6] node[auto,midway]{$b,c,d$} (q5);
\draw(q6) edge[loop,out=-95,in=-145,looseness=6] node[auto,pos=.38]{$b,c,d$} (q6);
%\useasboundingbox (q6) rectangle (q3);
\end{tikzpicture} \\
\tikzset{node distance=1.2cm,every loop/.style={min distance=.5cm,looseness=6}}%
\begin{tikzpicture}[scale=.95,every node/.style={transform shape}]
\node[state] (q0) {$q_0$};
\node[state] [right=of q0] (q1) {$q_1$};
\node[state] [right=of q1] (q2) {$q_2$};
\node[state] [below right=.5cm and .8cm of q2] (q3) {$q_3$};
\node[state] [below=of q2] (q4) {$q_4$};
\node[state] [left=of q4] (q5) {$q_5$};
\node[state] [left=of q5] (q6) {$q_6$};
\draw(q0) edge[dashed] (q5);
\node [above right=.3cm and -1cm of q0](label) {{\Large(1)}\quad $Q\setminus \{q_0\} \xrightarrow{b^{-1}} Q$; reduction by $q_5$};
\end{tikzpicture}\hfill
\begin{tikzpicture}[scale=.95,every node/.style={transform shape}]
\node[state] (q0) {$q_0$};
\node[state] [right=of q0] (q1) {$q_1$};
\node[state] [right=of q1] (q2) {$q_2$};
\node[state] [below right=.5cm and .8cm of q2] (q3) {$q_3$};
\node[state] [below=of q2] (q4) {$q_4$};
\node[state] [left=of q4] (q5) {$q_5$};
\node[state] [left=of q5] (q6) {$q_6$};
\draw(q0) edge (q5);
\draw(q1) edge[dashed] (q6);
\node[above right=.3cm and -1cm of q0](label) {{\Large(2)}\quad $Q\setminus \{q_1\} \xrightarrow{a^{-1}} Q \setminus \{q_0\}$; reduction by $\delta(q_5,a) = q_6$};
\end{tikzpicture}\\\vspace{.5cm}
\begin{tikzpicture}[scale=.95,every node/.style={transform shape}]
\node[state] (q0) {$q_0$};
\node[state] [right=of q0] (q1) {$q_1$};
\node[state] [right=of q1] (q2) {$q_2$};
\node[state] [below right=.5cm and .8cm of q2] (q3) {$q_3$};
\node[state] [below=of q2] (q4) {$q_4$};
\node[state] [left=of q4] (q5) {$q_5$};
\node[state] [left=of q5] (q6) {$q_6$};
\draw(q0) edge (q5);
\draw(q1) edge (q6);
\draw[dashed](q2) edge (q4);
\node[above right=.3cm and -1cm of q0](label) {{\Large(3)}\quad $Q\setminus \{q_2\} \xrightarrow{c^{-1}} Q$; reduction by $q_4$};
\end{tikzpicture}\hfill
\begin{tikzpicture}[scale=.95,every node/.style={transform shape}]
\node[state] (q0) {$q_0$};
\node[state] [right=of q0] (q1) {$q_1$};
\node[state] [right=of q1] (q2) {$q_2$};
\node[state] [below right=.5cm and .8cm of q2] (q3) {$q_3$};
\node[state] [below=of q2] (q4) {$q_4$};
\node[state] [left=of q4] (q5) {$q_5$};
\node[state] [left=of q5] (q6) {$q_6$};
\draw(q0) edge (q5);
\draw(q1) edge (q6);
\draw(q2) edge (q4);
\draw[dashed](q4) edge (q3);
\node[above right=.3cm and -1cm of q0](label) {{\Large(4)}\quad $Q\setminus \{q_4\} \xrightarrow{a^{-1}} Q \setminus \{q_2\}$; reduction by $\delta(q_4,a) = q_3$};
\end{tikzpicture}\\\vspace{.5cm}
\begin{tikzpicture}[scale=.95,every node/.style={transform shape}]
\node[state] (q0) {$q_0$};
\node[state] [right=of q0] (q1) {$q_1$};
\node[state,dashed] [below right=.25cm and 1.5cm of q1] (q2q3q4) {$q_2,q_3,q_4$};
\node[state] [left=of q4] (q5) {$q_5$};
\node[state] [left=of q5] (q6) {$q_6$};
\draw(q0) edge (q5);
\draw(q1) edge (q6);
\node[above right=.3cm and -1cm of q0](label) {{\Large(5)}\quad $Q\setminus \{q_3\} \xrightarrow{a^{-1}} Q \setminus \{q_4\}$; reduction by $\delta(q_3,a) = q_2$};
\end{tikzpicture}\hfill
\begin{tikzpicture}[scale=.95,every node/.style={transform shape}]
\node[state,dashed] (q0q5) {$q_0,q_5$};
\node[state] [right=of q0] (q1) {$q_1$};
\node[state] [below right=.25cm and 1.5cm of q1] (q2q3q4) {$q_2,q_3,q_4$};
\node[state] [left=of q5] (q6) {$q_6$};
\draw(q1) edge (q6);
\node[above right=.3cm and -1cm of q0](label) {{\Large(6)}\quad $Q\setminus \{q_5\} \xrightarrow{a^{-1}} Q \setminus \{q_1\}$; reduction by $\delta(q_6,a)=q_0$};
\end{tikzpicture}\\\vspace{.25cm}
\begin{tikzpicture}[scale=.95,every node/.style={transform shape}]
\node[state] (q0q5) {$q_0,q_5$};
\node[state,dashed] [right=of q0] (q1q6) {$q_1,q_6$};
\node[state] [right=1.5cm of q1] (q2q3q4) {$q_2,q_3,q_4$};
\node[above=.3cm of q1q6](label) {{\Large(7)}\quad $Q\setminus \{q_6\} \xrightarrow{a^{-1}} Q \setminus \{q_5\}$; reduction by $\delta(q_0,a)=q_1$};
\end{tikzpicture}
\caption{A possible sequence of computed reduction graphs for the not completely reachable automaton from~\Cref{fig:examples-notcompletely}~(right) (also shown above). In each step, the dashed line denotes the added edge or the dashed circle denotes the unified vertex.}\label{fig:ExampleReductionGraph}
\end{figure}
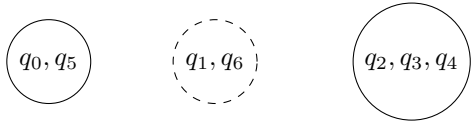

\begin{example}
\Cref{fig:ExampleReductionGraph} shows a possible sequence of reductions and the reduction graph after each iteration.
This also depicts a possible execution of the algorithm.
In the example, the minimal leaves and the letters to apply are chosen in lexicographical order.

In~(1), we pick leaf $\{q_0\}$ thus run with subset $Q \setminus \{q_0\}$.
Applying $a^{-1}$ yields $Q \setminus \{q_6\}$, which is the complement of another minimal leaf, hence we ignore it.
Next, we apply $b^{-1}$ and our set is properly extended to $Q$.
W have $|\delta^{-1}(q_5,b)| > 1$ and conclude by~\Cref{lem:ReductionByPredecessor} that every witness not containing $q_0$ also does not contain $q_5$.
Hence, we add an edge from $\vertexOf(q_0)$ to $\vertexOf(q_5)$.

In~(2), we pick subset $Q \setminus \{q_1\}$ and find out that by applying $a^{-1}$, we obtain $Q \setminus \{q_0\}$.
Now $\{q_0\}$ is not a minimal leaf, because of the edge added in~(1).
Hence, we can infer a reduction for $Q \setminus \{q_0\}$ from the reduction graph by following the path starting from $\vertexOf(q_0)$ to a vertex containing a state from $Q \setminus \{q_0\}$.
This reduction is $q_5$.
Then, we map back $q_5$ by the action of $a$, obtaining $q_6$, which is the state to remove from our original $Q \setminus \{q_1\}$.
We add an edge from $\vertexOf(q_1)$ to $\vertexOf(q_6)$.

In~(3), the situation is analogous to that in~(1) -- letter $c$ is properly extending for $Q \setminus \{q_2\}$, and we have $|\delta^{-1}(q_4,c)| > 1$.

In~(4), we first try reducing $Q \setminus \{q_3\}$.
However, $a^{-1}$ yields $Q \setminus \{q_4\}$ and the inverse actions of the three other letters yield the same set $Q \setminus \{q_3\}$.
They are complements of minimal leaves, so no reduction is found.
We try the next set $Q \setminus \{q_4\}$, and from this one, we get $Q \setminus \{q_2\}$ by $a^{-1}$, which can be reduced by $q_4$ from the reduction graph.
After mapping $q_4$ back by the action of $a$, we add an edge from $\vertexOf(q_4)$ to $\vertexOf(\delta(q_4,a)) = \vertexOf(q_3)$.

In~(5), we reduce $Q \setminus \{q_3\}$.
By $a^{-1}$ we get $Q \setminus \{q_4\}$, which is reduced by $q_3$ from the reduction graph.
We can remove $\delta(q_3,a) = q_2$ from $Q \setminus \{q_3\}$.
By adding an edge from $\vertexOf(q_3)$ to $\vertexOf(q_2)$, we get a cycle.
Hence, we replace its vertices with one with the union $\{q_2,q_3,q_4\}$ of their sets; it is a leaf, but not a minimal one.

In~(6), we pick $Q \setminus \{q_5\}$.
By applying $a^{-1}$, we get $Q \setminus \{q_1\}$, which is reduced by $q_6$ from the reduction graph.
Hence, the new edge is from $\vertexOf(q_5)$ to $\vertexOf(\delta(q_6,a)) = \vertexOf(q_0)$.
This results in a cycle with two vertices, which are replaced with the union $\{q_0,q_5\}$.

In~(7), we pick $Q \setminus \{q_6\}$, which is the unique minimal leaf.
By applying $a^{-1}$, we get $Q \setminus \{q_5\}$.
Now, we do not have an outgoing edge from $\vertexOf(q_5)$, but this vertex contains also $q_0 \in Q \setminus \{q_5\}$, thus we can reduce by $q_0$.
Hence, the new edge is from $\vertexOf(q_6)$ to $\vertexOf(\delta(q_0,a)) = \vertexOf(q_1)$.
As in the previous step, the cycle is replaced with the union $\{q_1,q_2\}$.

Finally, we end up with a reduction graph with three vertices, and the complements of these vertices cannot be reduced -- they are witnesses.
We know that by trying the inverse action of every letter, which either does not give a predecessor or yields the complement a minimal leaf, e.g., $\delta^{-1}(\{q_0,q_5\},a) = \{q_1,q_6\}$.
\end{example}

\begin{remark}
Our reduction graphs are similar to \emph{generalized Rystsov graphs} studied previously \cite{BCV2023CompletelyReachableInterplay}.
The latter are graphs on pairwise disjoint subsets of $Q$ that sum up to $Q$, where an edge $C \to C'$ denotes that there exists a properly extending word $w$ for $\overline{C}$ such that $|\delta^{-1}(p,w)|>1$ for some $p \in C'$. 
This implies that $Q \setminus C$ can be reduced for witness containment to $Q \setminus (C \cup C')$, hence they would be reduction graphs without the restriction that the out-degree is at most $1$.

However, their construction strategy is different.
They are defined in steps, where, in each step, all possible edges are added and then the condensation graph is taken (replacing strongly connected components with the unions of their vertices).
There is no known polynomial-time algorithm for building these graphs due to the difficulty of considering all words of a certain rank.

In contrast, for our reduction graphs, in every step, we add just one edge and always from a leaf, keeping the maximum out-degree $1$.
Further choosing the leaves to be minimal ones lets us use the current reduction graph to effectively find a new edge to add.
\end{remark}

%%%%%%%%%%%%%%%%%%%%%%%%%%%%%%%%%%%%%%%%%%%%%%%%%%%%%%%%%%%%%%%%%%%%%%%%%%%%%%%%%%%%%%%%%%%%%%%%%%%%%%%
\subsection{Formal and technical definition of a reduction graph}

For an automaton $\mathrsfs{A}=(Q,\Sigma,\delta)$, a \emph{reduction graph} is a directed forest $\mathcal{R}=(\mathcal{V}_\mathcal{R},\mathcal{E}_\mathcal{R})$, where the vertices in $\mathcal{V}_\mathcal{R}$ are pairwise disjoint subsets of $Q$ that altogether sum up to $Q$, and each vertex has out-degree either $0$ or $1$.
Hence, for each state $q \in Q$, there is a unique vertex containing it; we denote it by $\vertexOf(q)$.
A vertex without an outgoing edge is a \emph{leaf}.
Since the reduction graph is acyclic, it has at least one leaf, and from every vertex, there is exactly one leaf reachable.
For a vertex $V$, by $\nextVertex(V)$ we denote either the unique vertex $V'$ such that there is an edge from $V$ to $V'$, or {\none} if $V$ is a leaf.
The \emph{scope} of a vertex $V$ is the union of all the vertices reachable from $V$ up to the leaf.

The reduction graph is \emph{valid} (for the automaton) if, for every state $q \in Q$, every witness $W \subseteq Q \setminus \{q\}$ is also contained in $Q \setminus \scopeOf(\vertexOf(q))$.
Leaves of the smallest size in the reduction graph are of special interest; we call them \emph{minimal leaves}.

Finding the vertex of a state can be performed in $\O(1)$ time if together with the reduction graph we maintain an auxiliary array map $Q \to \mathcal{V}_\mathcal{R}$.
Technically, vertices are maintained through identifiers.
By maintaining another array map, for each vertex, we store the list of states in it, so we can access the list in $\O(1)$ time.

The space used by the reduction graph is in $\O(n)$:
We store an array map $Q \to V_\mathcal{R}$ for computing $\vertexOf$.
For each vertex, we also store a list of states -- their lengths altogether sum up to $n$.
Finally, for each vertex, we store one optional outgoing edge.

%%%%%%%%%%%%%%%%%%%%%%%%%%%%%%%%%%%%%%%%%%%%%%%%%%%%%%%%%%%%%%%%%%%%%%%%%%%%%%%%%%%%%%%%%%%%%%%%%%%%%%%
\subsection{Reduction retrieval from the reduction graph}

Given a subset $T \subsetneq Q$ and a valid reduction graph $\mathcal{R}$, we can effectively find a reduction of $T$ for witness containment, under a certain additional requirement.
To do this, we pick up any state $q \in \overline{T}$ and start searching from the vertex $V = \vertexOf(q)$.
The goal is to find any state $p \in \scopeOf(V) \cap T$, which can be removed from $T$ because we know from the reduction graph that all witnesses contained in $T \subseteq Q \setminus \{q\}$ are also contained in $T \setminus \scopeOf(V)$.
Hence, the requirement is that $\scopeOf(V) \cap T$ is nonempty.
Finding just one state $p$ to remove can be done fast (in linear time) by traversing the reduction graph, and it will be enough for our algorithm because $p$ after remapping will serve as a representant of a vertex to which we need to add an edge.

\Cref{alg:GetImpliedReduction} shows function \Call{GetImpliedReduction}{}.
It just follows the path starting from an initially chosen vertex $V$ and checks all the states in the lists of vertices.
Here, we need $T$ given in the array form to check the state membership in constant time.

\begin{algorithm}[htb]
\caption{Fast reduction retrieval for witness containment using the reduction graph.}\label{alg:GetImpliedReduction}
\begin{algorithmic}[1]
\Input A reduction graph $\mathcal{R}$ for an automaton $\mathrsfs{A}$, and a nonempty set $T \subsetneq Q$ (given in the array form).
\Require $\mathcal{R}$ is valid for $\mathrsfs{A}$ and $T \cap \scopeOf(V)$ is nonempty for every vertex $V$.
\Output A state $p \in T$.
\Ensure Every witness of $\mathrsfs{A}$ contained in $T$ is also contained in $T \setminus \{p\}$.
\Complexity $\O(n)$ time and $\O(1)$ working space (\Cref{lem:GetImpliedReduction}).
\Function{GetImpliedReduction}{$\mathcal{R}$, $T$}
\State $q \gets$ any state from $\overline{T}$
\State $V \gets \vertexOf(q)$
\While{\textbf{true}}
\ForAll{$p \in V$}
  \If{$p \in T$}
    \State \Return $p$
  \EndIf
\EndFor
\State $V \gets \nextVertex(V)$
\State \Assert{$V \ne \none$} \Comment{Cannot happen due to the requirement}
\EndWhile
\EndFunction
\end{algorithmic}
\end{algorithm}

\begin{lemma}\label{lem:GetImpliedReduction}
\Call{GetImpliedReduction}{} from~\Cref{alg:GetImpliedReduction} is correct and works in $\O(n)$ time and $\O(1)$ working space (not counting the arguments).
\end{lemma}
\begin{proof}
\noindent\textit{Correctness}:
By the requirement, $T \cap \scopeOf(V)$ is nonempty in particular for the initially chosen vertex $V$.
The function subsequently checks all states $p \in \scopeOf(V)$, hence it finally finds a state $p$ such that $p \in T$.
Since $\mathcal{R}$ is valid, all witnesses contained in $T \subseteq Q \setminus \{q\}$ are also contained in $T \cap (Q \setminus \{q,p\})$ thus in $T \setminus \{p\}$.

\noindent\textit{Running time}:
Line~2 works in $\O(n)$ time and line~3 works in $\O(1)$ time.
The main \textbf{while} loop works in $\O(n-|T|)$ time as follows.
Iteration over all $p \in V$ takes $\O(|V|)$ time:
taking a next state $p$ from $V$ works in constant time by traversing its list of states; the condition $p \in T$ in line~6 is checked in $\O(1)$ time by an array lookup to $T$.
Line~8 works in $\O(1)$ time by using the array map storing outgoing edges.
Since the vertices are disjoint and nonempty, all iterations of the main loop take at most $\O(n)$ time.

\noindent\textit{Running space}:
Not counting the arguments, we need to store only $r$, $V$ (identifier), and $p$ together with an index for iteration over the list of states of $V$.
These fit in $\O(1)$ working space.
\end{proof}

%%%%%%%%%%%%%%%%%%%%%%%%%%%%%%%%%%%%%%%%%%%%%%%%%%%%%%%%%%%%%%%%%%%%%%%%%%%%%%%%%%%%%%%%%%%%%%%%%%%%%%%
\subsection{Finding a reduction with the help of the reduction graph}

\Cref{alg:FindReduction} shows function \Call{FindReduction}{}, which is responsible for finding a new reduction for the complement set of a leaf.
For the complements of leaves, a reduction cannot be obtained directly from the reduction graph, because the scope of a leaf is the leaf itself thus it is disjoint from its complement.

The function takes a minimal leaf $V$ and searches for a reduction for the set $\overline{V}$.
It does not need a queue (as \Call{IsWitness}{} from~\Cref{alg:IsWitness}) and only considers one application of the inverse action of every letter.
If a larger predecessor is found, then the reduction is directly derived by~\Cref{lem:ReductionByPredecessor} since the letter is a properly extending word.
The second case is to obtain a predecessor $T$ that is not the complement of a minimal leaf.
Then, since also $n-|T|$ is the size of minimal leaves, we have the guarantee that for every vertex $V'$, $T \cap \scopeOf(V')$ is nonempty, hence we can use \Call{GetImpliedReduction}{}.

\begin{algorithm}[htb!]
\caption{Finding a reduction for witness containment with the help of the reduction graph.}\label{alg:FindReduction}
\begin{algorithmic}[1]
\Input An $n$-state automaton $\mathrsfs{A}=(Q,\Sigma,\delta)$ with the preprocessed data from~\Cref{lem:ComplementList}, a reduction graph $\mathcal{R}$ for $\mathrsfs{A}$, and a minimal leaf $V$ in $\mathcal{R}$.
\Require $\mathcal{R}$ is valid for $\mathrsfs{A}$.
\Output A state $p \in \overline{V}$ or \none.
\Ensure If there exists an $a$-predecessor of $\overline{V}$, where $a \in \Sigma$, that is either larger than $\overline{V}$ or is not the complement of a minimal leaf in $\mathcal{R}$, then the function returns a state $r \in \overline{V}$ such that all witnesses contained in $\overline{V}$ are also contained in $\overline{V} \setminus \{r\}$.
Otherwise, the function returns \none.
\Complexity If a state is returned, $\O(|\Sigma|\cdot |V| + n)$ time; otherwise, $O(|\Sigma|\cdot |V|)$ time. The working space is $\O(n)$. (\Cref{lem:FindReduction})
\Function{FindReduction}{$\mathrsfs{A},\mathcal{R},V$}
\ForAll{$a \in \Sigma$}
  \If{$\overline{V}$ has $a$-predecessor}
    \State $T = \delta^{-1}(\overline{V},a)$
    \If{$|T| > |\overline{V}|$}
      \State $r \gets$ any state $r \in \overline{V}$ such that $|\delta^{-1}(r,a)| > 1$
      \State \Return $r$
    \EndIf
    \State \Assert{$|T| = |\overline{V}|$}
    \If{$\overline{T}$ is not the set of a minimal leaf in $\mathcal{R}$}
      \State $p \gets \Call{GetImpliedReduction}{\mathcal{R},T}$
      \State \Return $\delta(p,a)$
    \EndIf
  \EndIf
\EndFor
\State \Return \textbf{none}
\EndFunction
\end{algorithmic}
\end{algorithm}

\begin{lemma}\label{lem:FindReduction}
\Call{FindReduction}{} from~\Cref{alg:FindReduction} is correct and works in $\O(|\Sigma|\cdot |V|)$ time if the answer is $\none$, and in $\O(|\Sigma|\cdot |V| + n)$ time if the answer is a state.
The working space is in $\O(n)$ (not counting the arguments).
\end{lemma}
\begin{proof}
\noindent\textit{Correctness}:
There are three possible exits from the function.

In line~7, letter $a$ is a properly extending word for $\overline{V}$, thus we know that the reduction by $r$ is correct by~\Cref{lem:ReductionByPredecessor}, i.e., all witnesses contained in $\overline{V}$ are also contained in $\overline{V} \setminus \{r\}$.

For line~11, we need to show that all witnesses in $\overline{V}$ are contained in $\overline{V} \setminus \{\delta(p,a)\}$.
In line~8, we know that $T$ has the same size as $\overline{V}$ because a predecessor cannot be smaller than the set and we check if it is larger in line~5.
Hence, $\overline{T}$ has the size of the minimal leaves in $\mathcal{R}$. 
In line~10, we check if $\overline{T}$ is not a minimal leaf.
Since the scope of every vertex that is not a minimal leaf is larger than $\overline{T}$, this means that the scope of such a vertex contains a state that is in $T$.
Hence, \Call{GetImpliedReduction}{} meets the requirement, so provides a correct answer, and we know that all witnesses contained in $T$ are contained in $T \setminus \{p\}$.
Now, suppose for a contradiction that there is a witness $W \subseteq \overline{V}$ that contains $\delta(p,a)$.
Since $T$ is $a$-predecessor of $\overline{V}$, $W' = \delta^{-1}(W,a)$ is $a$-predecessor of $W \subseteq \overline{V}$, and $W' \subseteq T$.
Since $\delta(p,a) \in W$, we also have $p \in W'$.
By~\Cref{cor:Witness}~(2), $W'$ is also a witness, but since it is contained in $T$, it must be contained also in $T \setminus \{p\}$, which yields a contradiction.
It follows that all witnesses contained in $\overline{V}$ are contained in $\overline{V} \setminus \{\delta(p,a)\}$, thus the outcome of the function is correct.

When the function reaches the final case in~line~12, all predecessors of $\overline{V}$ by one letter were checked, so they all are of the same size $|\overline{V}|$ and are the complements of minimal leaves.
Thus, the outcome is as specified.

\noindent\textit{Running time}:
When $V$ is stored as a list of states, thus as a list of states in the complement of $\overline{V}$, lines~3--4 can be done in $\O(n-|\overline{V}|) = \O(|V|)$ time by~\Cref{lem:ComplementList}.
This list is stored for $V$ in the reduction graph.
The condition in line~5 is done in $\O(1)$ time.
The condition in line~9 can be checked in $\O(n-|T|) = \O(|V|)$ time by examining each $\vertexOf(q)$ for all $q \in \overline{T}$, looking whether it is the same vertex for all letters $a$ and its size equals $|\overline{V}|$, which is the size of the minimal leaves.
Thus, the whole \textbf{for} loop in line~3 works in $\O(|\Sigma|\cdot |V|)$ time, with the possible exception of the last iteration when a state is returned.
Lines~6--7 are done in $\O(n)$ by examining each $p \in \overline{V}$; note that we cannot do this in $\O(|V|)$ time.
Lines~10--11 are also done in $\O(n)$ time, including the conversion of $T$ to the array form for calling \Call{GetImpliedReduction}{}.
One of these cases is executed if only if a state is returned -- they are responsible for the $+n$ component in the running time formula.

\noindent\textit{Working space}:
In the working space, we need to store only $V$, $a$, $T$, $r$, and $p$, and a constant number of indices for iterations, which fit in $\O(|V|)$ when the sets are stored as lists of states in their complements.
If there is a call to \Call{GetImpliedReduction}{}, then we need $\O(n)$ space for converting and passing $T$ in the array form.
\end{proof}

%%%%%%%%%%%%%%%%%%%%%%%%%%%%%%%%%%%%%%%%%%%%%%%%%%%%%%%%%%%%%%%%%%%%%%%%%%%%%%%%%%%%%%%%%%%%%%%%%%%%%%%
\subsection{Updating the reduction graph}

When we know a minimal leaf $V$ and a state $r \in \overline{V}$ such that every witness contained in $\overline{V}$ is contained in $Q \setminus (V \cup \{r\})$, we need to update the reduction graph with this new knowledge.
From the reduction graph, we know further that every witness in $\overline{V}$ must be contained in $Q \setminus (V \cup V')$, where $V' = \vertexOf(r)$.
Hence, we add the edge $V \to V'$.
Since there were no edges outcoming from $V$, the obtained graph still has its maximum out-degree $1$.
Now, there are two cases.
If there is no cycle, the graph is acyclic and we are done.
If there is a cycle, then we replace all the vertices from the cycle with one vertex $U$ which is the union of them.
All edges that were incoming to a vertex in the cycle are redirected to $U$.
Note that no edges were outgoing from the cycle to a vertex outside it, hence $U$ becomes a leaf.

Performing such an update can be done in $\O(n)$ time and space.
We add an edge in constant time.
Then we check if there is a cycle simply by following the path from $V$ and checking if we encounter an already visited vertex.
If so, we replace all these visited vertices with the union vertex $U$, and then we iterate over all edges to redirect those going to a vertex in the cycle.
In this case, also the list of states for $U$ is computed by concatenating the lists for the vertices in the cycle.

We denote the described updating procedure by $\Call{UpdateReductionGraph}{\mathcal{R},V,V'}$, which takes two vertices to add an edge.

%%%%%%%%%%%%%%%%%%%%%%%%%%%%%%%%%%%%%%%%%%%%%%%%%%%%%%%%%%%%%%%%%%%%%%%%%%%%%%%%%%%%%%%%%%%%%%%%%%%%%%%
\subsection{Adding a reduction}

We build one more auxiliary function for performing one step of finding a reduction followed by updating the reduction graph.
\Call{AddReduction}{} is shown in~\Cref{alg:AddReduction}.
The function iterates over all minimal leaves in $\mathcal{R}$ and uses \Call{FindReduction}{} to find a reduction -- a state to remove.
When it is found, it calls \Call{UpdateReductionGraph}{}.

The function works in $\O(|\Sigma|\cdot n)$ time despite multiple calls to \Call{FindReduction}{}, which is due to maintaining sets in the form of a list of states in the complement and because leaves are disjoint.

\begin{algorithm}[htb!]
\caption{Adding a reduction to the reduction graph.}\label{alg:AddReduction}
\begin{algorithmic}[1]
\Input An $n$-state automaton $\mathrsfs{A}=(Q,\Sigma,\delta)$ with the preprocessed data from~\Cref{lem:ComplementList} and a reduction graph $\mathcal{R}$ for $\mathrsfs{A}$.
\Require The reduction graph $\mathcal{R}$ is valid for $\mathrsfs{A}$.
\Output \textbf{true} if $\mathcal{R}$ has been updated. \textbf{false} if the complements of all minimal leaves in $\mathcal{R}$ are maximal witnesses.
\Ensure $\mathcal{R}$ remains valid.
\Complexity $\O(|\Sigma| \cdot n)$ time and $\O(n)$ working space (\Cref{lem:AddReduction}).
\Function{AddReduction}{$\mathrsfs{A}$,$\mathcal{R}$}
\ForAll{$V \gets $ a minimal leaf in $\mathcal{R}$}
  \State $r \gets \Call{FindReduction}{\mathrsfs{A},\mathcal{R},V}$
  \If{$q \neq \none$}
    \State $\Call{UpdateReductionGraph}{\mathcal{R},V,\vertexOf(r)}$
    \State \Return \true
  \EndIf
\EndFor
\State \Return \false
\EndFunction
\end{algorithmic}
\end{algorithm}

\begin{lemma}\label{lem:AddReduction}
Function~\Call{AddReduction}{} from~\Cref{alg:AddReduction} is correct and works in $\O(|\Sigma|\cdot n)$ time and $\O(n)$ working space (not counting the automaton).
\end{lemma}
\begin{proof}
\noindent\emph{Correctness}:
The correctness follows from the specification of~\Call{FindReduction}{} and~\Call{UpdateReductionGraph}{}.

When the function returns \true, the call to \Call{FindReduction}{} for the last chosen minimal leaf $V$ has returned a state $r$ such that every witness in $\overline{V}$ is also contained in $\overline{V} \setminus \{r\}$.
Since the reduction graph is valid, every such witness is also contained in $\overline{V} \setminus \vertexOf(r)$.
Hence, the reduction graph remains valid after updating by \Call{UpdateReductionGraph}{} with the edge from $V$ to $\vertexOf(r)$.

When the function returns \false, the reduction graph is unmodified and thus trivially remains valid.
In this case, for all minimal leaves $V$, the call to \Call{FindReduction}{} returned \none.
Hence, for each such $V$, all predecessors of $\overline{V}$ by one letter are also complements of minimal leaves.
Every such $\overline{V}$ does not have a larger predecessor, and the complements of its predecessors are pairwise disjoint since they are some vertices in the reduction graph.
Hence, all $\overline{V}$ are witnesses.
Furthermore, they are maximal witnesses, because every larger subset could be reduced using the reduction graph.
Indeed, let $S \subsetneq Q$ be a larger subset than the witness $\overline{V}$, where $V$ is a minimal leaf.
Choose any $q \in \overline{S}$.
The scope of every vertex has its size of at least $|V|$, hence in particular $\scopeOf(\vertexOf(q))$.
Since $\overline{S}$ is smaller than $|V|$, the intersection $S \cap \scopeOf(\vertexOf(q))$ is nonempty, hence, we can obtain for $S$ a state to reduce for witness containment by \Call{GetImpliedReduction}{}.

\noindent\emph{Running time:}
Consider all iterations of the \textbf{for} loop (line~2) where $q = \none$, thus all iterations except possibly the last one.
Then each call to \Call{FindReduction}{} works in $\O\left(|\Sigma|\cdot |V|\right)$ time by~\Cref{lem:FindReduction}.
Since the sum of the sizes of all considered sets $V$ is at most $n$, the running time of all these iterations does not exceed $\O(|\Sigma|\cdot n)$.
In the possible last iteration with $q \ne \none$, \Call{FindReduction}{} works in $\O(|\Sigma|\cdot |V| + n)$ time, and \Call{UpdateReductionGraph}{} works in $\O(n)$ time.
Altogether, the time complexity is in $\O(|\Sigma|\cdot n)$.

\noindent\emph{Working space}:
We store $\mathcal{R}$ (it can be modified), $V$, $q$, $\vertexOf(q)$, which altogether fit in $\O(n)$.
Also, the calls to \Call{FindReduction}{} and \Call{UpdateReductionGraph} take $\O(n)$ working space.
\end{proof}

%%%%%%%%%%%%%%%%%%%%%%%%%%%%%%%%%%%%%%%%%%%%%%%%%%%%%%%%%%%%%%%%%%%%%%%%%%%%%%%%%%%%%%%%%%%%%%%%%%%%%%%
\subsection{Final algorithm}

The final algorithm is shown in~\Cref{alg:FindMaximalWitness}.
We start from the preprocessing of the automaton from~\Cref{lem:ComplementList}, which is crucial for computing predecessors in~\Call{FindReduction}{} in time $\O(|V|)$ instead of $\O(n)$.
We initialize the reduction graph: in the beginning, it has $n$ vertices and no edges; for each state $q \in Q$, there is one vertex $V = \{q\}$.
In the main loop, we call \Call{AddReduction}{} until either it fails (returns \false), meaning that we encountered a witness, or the reduction graph finally becomes a single vertex, meaning that there are no witnesses.

The running time is the consequence of the number of iterations where the reduction graph is updated, which is at most $2n-2$, and of the running time $\O(|\Sigma|\cdot n)$ of \Call{AddReduction}{} called in one iteration.

\begin{algorithm}[htb!]
\caption{The algorithm for finding a maximal witness.}
\label{alg:FindMaximalWitness}
\begin{algorithmic}[1]
\Input An $n$-state automaton $\mathrsfs{A}=(Q,\Sigma,\delta)$.
\Output \textbf{none} if $\mathrsfs{A}$ is completely reachable; a maximal witness otherwise.
\Complexity $\O(|\Sigma| \cdot n^2)$ time and $\O(|\Sigma|\cdot n)$ space (\Cref{thm:FindMaximalWitness}).
\Function{FindMaximalWitness}{$\mathrsfs{A}$}
\State Preprocess $\mathrsfs{A}$ as in~\Cref{lem:ComplementList}
\State $\mathcal{R} \gets$ empty graph with $n$ vertices with assigned singletons of states
\While{$\mathcal{R}$ has at least $2$ vertices}
  \If{\textbf{not} $\Call{AddReduction}{\mathrsfs{A},\mathcal{R}}$}
    \State $V \gets$ any minimal leaf in $\mathcal{R}$
    \State \Return $Q \setminus V$ \Comment{A maximal witness}
  \EndIf
\EndWhile
\State \Return \none \Comment{No witnesses -- $\mathrsfs{A}$ is completely reachable}
\EndFunction
\end{algorithmic}
\end{algorithm}

\begin{lemma}\label{lem:UpdatesReductionGraph}
There are at most $2n-2$ updates (\Call{UpdateReductionGraph}{}) of the reduction graph until it has one vertex.
\end{lemma}
\begin{proof}
One update of the reduction graph either (1) adds an edge, or, (2) after adding an edge, unifies at least two vertices.

In~(2), the number of vertices is decreased, hence, there are at most $n-1$ such updates.

In~(1), adding an edge decreases the number of weakly connected components (maximal subsets of vertices that are mutually reachable when the edges are considered undirected).
Furthermore, in case~(2), the number of weakly connected components is preserved, because the edge is added between vertices in the same weakly connected component, and unifying the cycle trivially does not change their number.
In the beginning, there are $n$ weakly connected components, hence there are at most $n-1$ updates of case~(1).
\end{proof}

\begin{theorem}\label{thm:FindMaximalWitness}
Function~\Call{FindMaximalWitness}{} is correct and works in $\O(|\Sigma|\cdot n^2)$ time and $\O(|\Sigma|\cdot n)$ space.
\end{theorem}
\begin{proof}
\noindent\emph{Correctness}:
Note that the reduction graph $\mathcal{R}$ is kept valid.
The empty reduction graph is trivially valid, and we know that \Call{AddReduction}{} preserves its validity (\Cref{lem:AddReduction}).

If \Call{AddReduction}{} returns \false, we know that the complements of minimal leaves in $\mathcal{R}$ are maximal witnesses, hence the answer is correct.

When $\mathcal{R}$ has $1$ vertex, this means that every witness could be reduced to the empty set, i.e., for every subset $S \subsetneq Q$, every witness contained in $S$ is contained in $\emptyset$, thus there are no witnesses.

\noindent\emph{Running time}:
Line~1 works in $\O(|\Sigma|\cdot n)$ time by~\Cref{lem:ComplementList}.
Lines~3, 6, and~7 work in $\O(n)$ time and they are executed (at most) once.
Checking the condition in line~4 takes constant time.
\Call{AddReduction}{} in line~5 works in $\O(|\Sigma|\cdot n)$ time.

By~\ref{lem:UpdatesReductionGraph}, there are at most $2n-2$ updates of the reduction graph before the condition in line~4 becomes false.
Hence, this bounds the number of iterations.

Summarizing, the main loop works in $\O(|\Sigma|\cdot n^2)$ time, which dominates the time complexity of the function.

\noindent\emph{Running space}:
The space is dominated by storing the automaton $\mathrsfs{A}$ and its preprocessed structures, which both take $\O(|\Sigma|\cdot n)$ space.
Besides these, we need only $\O(n)$ space for the reduction graph, $V$ and the returned set, and the calls to \Call{AddReduction}{}.
\end{proof}

\begin{remark}
After a slight modification, \Call{FindMaximalWitnessOptimized}{} can return all maximal witnesses as well.
When \Call{AddReduction}{} returns {\false}, the complements of all minimal leaves in the final reduction graph are maximal witnesses.
\end{remark}

%%%%%%%%%%%%%%%%%%%%%%%%%%%%%%%%%%%%%%%%%%%%%%%%%%%%%%%%%%%%%%%%%%%%%%%%%%%%%%%%%%%%%%%%%%%%%%%%%%%%%%%
%%%%%%%%%%%%%%%%%%%%%%%%%%%%%%%%%%%%%%%%%%%%%%%%%%%%%%%%%%%%%%%%%%%%%%%%%%%%%%%%%%%%%%%%%%%%%%%%%%%%%%%
%%%%%%%%%%%%%%%%%%%%%%%%%%%%%%%%%%%%%%%%%%%%%%%%%%%%%%%%%%%%%%%%%%%%%%%%%%%%%%%%%%%%%%%%%%%%%%%%%%%%%%%
\section{A quadratic upper bound on reaching threshold}\label{sec:Bounds}

We show a quadratic upper bound on the length of the reaching threshold of a subset, provided that all subsets of the same or larger size are reachable.
Hence, in particular, this applies to the class of completely reachable automata and proves a weaker Don's conjecture for this class.

For a given nonempty proper subset $S \subsetneq Q$, we can find a short properly extending word.
Hence, for a completely reachable automaton, using the well-known \emph{extension} method (e.g., \cite{Volkov2008Survey,Volkov2022Survey}), we can construct a reaching word starting from $S$ and iteratively increasing its preimage by at least one by applying the inverse action of the found word.
Finally, we obtain the preimage equal to $Q$ in most $n-|S|$ iterations.

\subsection{Finding short properly extending words}

Our method for finding a short properly extending word is based on function \Call{IsWitness}{} from~\Cref{alg:IsWitness}, further developed suitably, and the complement-intersecting technique, whose core is described in~\Cref{lem:UnionPredecessor}.

\begin{algorithm}[htb!]
\caption{Finding a short properly extending word.}
\label{alg:FindShortProperlyExtendingWord}
\begin{algorithmic}[1]
\Input An $n$-state automaton $\mathrsfs{A}=(Q,\Sigma,\delta)$ and a nonempty $S \subsetneq Q$.
\Require $S$ and all subsets of $Q$ of size $> |S|$ are reachable.
\Output A properly extending word for $S$.
\Function{FindShortProperlyExtendingWord}{$\mathrsfs{A},S$}
\State $\mathit{Trace} \gets \Call{EmptyMap}{ }$ \Comment{For a processed set, it stores how this set was obtained; for not yet processed sets, it gives $\textbf{none}$}
\State $\mathit{Process} \gets \Call{EmptyFifoQueue}{ }$
\State $\mathit{Trace}[S] \gets \varepsilon$
\State $\mathit{Process}.\Call{Push}{S}$
\While{\textbf{true}}
  \State \textbf{assert}($\textbf{not}\ \mathit{Process}.\Call{IsEmpty}{ }$) \Comment{Otherwise $S'$ must be a witness}
  \State $T \gets \mathit{Process}.\Call{Pop}{ }$
  \If{there is $T'$ such that $\mathit{Trace}[T'] \neq \textbf{none}$ and $T \subsetneq T \cup T' \subsetneq Q$}
    \State $\mathit{Trace}[T \cup T'] \gets (T, T')$
    \State $\mathit{Process}.\Call{Clear}{ }$ \Comment{Continue only for the new set}
    \State  $\mathit{Process}.\Call{Push}{T \cup T'}$
  \Else
    \ForAll{$a \in \Sigma$}
      \If{$T$ has $a$-predecessor}
        \State $T' \gets \delta^{-1}(T,a)$
        \If{$\mathit{Trace}[T'] = \textbf{none}$} \Comment{A new set}
          \State $\mathit{Trace}[T'] \gets a$ \Comment{$\delta(T',a)=T$}
          \If{$|T'| > |T|$}
            \State \Return $\Call{Reconstruct}{\mathrsfs{A},T',S,\mathit{Trace}}$
          \EndIf
          \State $\mathit{Process}.\Call{Push}{T'}$
        \EndIf
      \EndIf
    \EndFor
  \EndIf
\EndWhile
\EndFunction
\Function{Reconstruct}{$\mathrsfs{A},C,S,\mathit{Trace}$} \Comment{Reconstruct word $w$ such that $\delta^{-1}(S,w)=C$}
\State $w \gets \varepsilon$ \Comment{Start with the empty word}
\While{$C \neq S$}
  \State \textbf{assert}($\mathit{Trace}[C] \neq \textbf{none}$ and $\mathit{Trace}[C] \neq \varepsilon$)
  \If{$\mathit{Trace}[C]$ is a letter}
    \State $a \gets \mathit{Trace}[C]$
    \State $w \gets wa$
    \State $C \gets \delta(C,a)$
  \Else
    \State $(T,T') \gets \mathit{Trace}[C]$ \Comment{$w$ properly extends $T \cup T'$ so either $T$ or $T'$}
    \If{$|\delta^{-1}(T,w)| > |T|$} \Comment{$w$ properly extends $T$}
      \State $C \gets T$
    \Else \Comment{$w$ properly extends $T'$} 
      \State $C \gets T'$
    \EndIf
  \EndIf
\EndWhile
\State \Return $w$
\EndFunction
\end{algorithmic}
\end{algorithm}

\Cref{alg:FindShortProperlyExtendingWord} shows two functions.
\Call{FindShortProperlyExtendingWord}{} is the main one, which calls \Call{Reconstruct}{} at the end.

The idea is a breadth-first search on predecessors of the given subset $S$.
We search for a direct larger predecessor or a nonempty complements intersection of the current set and a previously visited one.
The queue $\mathit{Process}$ stores the next sets to consider.
All visited subsets are stored for an intersection of their complements in the map $\mathit{Trace}$, which for a visited set, stores either the letter that was used to obtain this set as a maximal predecessor from a previous set, or in the second case, a pair of smaller sets whose union is the given set.
This map is also finally used to reconstruct the properly extending word.

The next set $T$ is taken from the queue in line~8.
Then we check it for a nonempty complement intersection in line~9.
If this happens, the information on how the larger set was obtained, i.e., the two sets of the union, is stored in $\mathit{Trace}$ (line~10).
The search is reset to continue only from this union (lines~11--12).
However, all previously considered sets are kept as keys in $\mathit{Trace}$ and still participate in the search for a complement intersection (line~9).

If the complement intersection case does not occur, we compute all the preimages of $T$ by one letter (lines~14--21).
If the same preimage has been already visited, it is skipped (line~17).
If a larger preimage is found (line~19), then the search ends and the reconstruction phase begins (line~20).
Because of the input requirements, all visited sets that we encounter are reachable, so finally, a larger predecessor must be found.

The reconstruction process is fully based on map $\mathit{Trace}$.
The word $w$ is the currently reconstructed part of the properly extending word for $S$, and the set $C$ is the current set from which we must reach $S$.
In each iteration of the main \textbf{while} loop (line~24), one of the two cases occurs.
If the current set was obtained as the preimage of a previous set under the action of a letter, we apply this letter to $C$ and append it to $w$ (lines~27--29).
In the second case, we choose the set of the two which formed the larger union (lines~31--35), depending on for which one the current word is properly extending.
Finally, the function ends when $C = S$, which means that $w$ is the reconstructed properly extending word for $S$.

\begin{example}
For the right automaton from~\Cref{fig:examples-completely} (completely reachable) and $S = \{q_0\}$, \Call{FindShortProperlyExtendingWord}{} returns the word $ca^2$ (assuming alphabetical order on $\Sigma$ in line~14 and lexicographical order on subsets in line~9). 
\end{example}
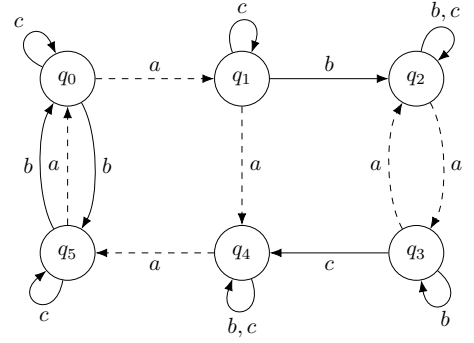
\begin{minipage}{0.60\textwidth}
\setlength{\parindent}{1em}
\setlength{\parskip}{1em}
\renewcommand{\arraystretch}{1.2}
$\begin{array}{l|l|l|l}
\text{Iteration} & \text{Set $T$} & \text{Added entries to $\mathit{Trace}$} \\\hline
Init &             & \mathit{Trace}[\{q_0\}]=\varepsilon \\
1    & S = \{q_0\} & \mathit{Trace}[\{q_5\}]=a \\
2    & \{q_5\}     & \mathit{Trace}[\{q_0,q_5\}] = (\{q_5\},\{q_0\}) \\
3    & \{q_0,q_5\} & \mathit{Trace}[\{q_4,q_5\}]=a \\
4    & \{q_4,q_5\} & \mathit{Trace}[\{q_0,q_4,q_5\}]= (\{q_4,q_5\},\{q_0\}) \\
\multirow{2}{*}{5} & \multirow{2}{*}{$\{q_0,q_4,q_5\}$} & \mathit{Trace}[\{q_1,q_4,q_5\}] = a, \\
                   &                 & \mathit{Trace}[\{q_0,q_3,q_4,q_5\}] = c \\
\end{array}$
\end{minipage}
\begin{minipage}{0.35\textwidth}
\hspace{.1cm}
\scalebox{0.83}{\begin{tikzpicture}[node distance=1.9cm,scale=1,every node/.style={transform shape}]
%\tikzset{every edge/.style={draw,->,>=stealth,auto,semithick}}
\tikzset{every loop/.style={min distance=.5cm,looseness=6}}
\node[state] (q0) {$q_0$};
\node[state] [right=of q0] (q1) {$q_1$};
\node[state] [right=of q1] (q2) {$q_2$};
\node[state] [below=of q2] (q3) {$q_3$};
\node[state] [left=of q3] (q4) {$q_4$};
\node[state] [left=of q4] (q5) {$q_5$};
\draw(q1) edge node[auto,midway]{$b$} (q2);
\draw(q3) edge node[auto,midway]{$c$} (q4);
\draw(q0) edge[dashed] node[auto,midway]{$a$} (q1);
\draw(q1) edge[dashed] node[auto,midway]{$a$} (q4);
\draw(q4) edge[dashed] node[auto,midway]{$a$} (q5);
\draw(q5) edge[dashed] node[auto,midway]{$a$} (q0);

\draw(q0) edge[looseness=.7,out=-60,in=60] node[auto,midway]{$b$} (q5);
\draw(q5) edge[looseness=.7,bend left] node[auto,midway]{$b$} (q0);

%\draw(q2) edge[dashed,looseness=.7,out=-65,in=65] (q3);
\draw(q2) edge[dashed,looseness=.7,out=-60,in=60] node[auto,midway]{$a$} (q3);
\draw(q3) edge[dashed,looseness=.7,bend left] node[auto,midway]{$a$} (q2);

\draw(q0) edge[loop,out=155,in=115,looseness=6] node[auto,pos=.6]{$c$} (q0);
\draw(q1) edge[loop,out=110,in=70,looseness=6] node[auto,midway]{$c$} (q1);
\draw(q2) edge[loop,out=80,in=40,looseness=6] node[above,midway]{$b,c$} (q2);
\draw(q3) edge[loop,out=-40,in=-80,looseness=6] node[below,midway]{$b$} (q3);
\draw(q4) edge[loop,out=-70,in=-110,looseness=6] node[auto,midway]{$b,c$} (q4);
\draw(q5) edge[loop,out=-100,in=-140,looseness=6] node[auto,pos=.4]{$c$} (q5);
\end{tikzpicture}} 
\end{minipage}
\begin{proof}[Proof of the example]
The execution is summarized in the above table and described below.
\begin{itemize}
\item Before the main loop, we just add $\mathit{Trace}[S] = \varepsilon$.
\item 1st iteration: $T = S = \{q_0\}$ and its preimages are computed. Only $a^{-1}$ gives a new preimage, which $\{q_5\}$; $\delta^{-1}(\{q_0\},b) = \{q_0\}$ is skipped.
\item 2nd iteration: The case from line~9 occurs. $\{q_5\}$ has a nonempty complement intersection with (and is not contained in) previously visited $\{q_0\}$, so the search continues from the union $\{q_5\} \cup \{q_0\}$; both these sets are stored in $\mathit{Trace}[\{q_0,q_5\}]$.
\item 3rd iteration: Only $a^{-1}$ gives a new preimage from $\{q_0,q_5\}$, which is $\{q_4,q_5\}$.
\item 4th iteration: $\{q_4,q_5\}$ has a nonempty complement intersection with (and is not contained in) previously visited $\{q_0\}$. Alternatively, $\{q_0,q_5\}$ instead of $\{q_0\}$ could be chosen as well. The search continues from the union $\{q_4,q_5\} \cup \{q_0\}$.
\item 5th iteration: Two letters give a new preimage: $\delta^{-1}(\{q_0,q_4,q_5\},a) = \{q_1,q_4,q_5\}$ and $\delta^{-1}(\{q_0,q_4,q_5\},c) = \{q_0,q_3,q_4,q_5\}$. The latter is a larger predecessor, so we start the reconstruction phase.
\item Reconstruction: Starting from $C = \{q_0,q_3,q_4,q_5\}$, we have $\mathit{Trace}[C] = c$.
Then $C = \{q_0,q_4,q_5\}$ and $\mathit{Trace}[C] = (\{q_4,q_5\},\{q_0\})$; only the first set is properly extended by the current $w = c$, so the next $C = \{q_4,q_5\}$.
We have $\mathit{Trace}[C] = a$ and $C$ becomes $\{q_0,q_5\}$.
Then $\mathit{Trace}[\{q_0,q_5\}] = (\{q_5\},\{q_0\})$ and the first set is properly extended by the current $w = ca$.
Finally, $\mathit{Trace}[\{q_5\}]=a$ and $C = \{q_0\} = S$.
The reconstructed word $w$ is $c a^2$.
\end{itemize}
\end{proof}

\begin{lemma}
Function~\Call{FindShortProperlyExtendingWord}{} is correct.
\end{lemma}
\begin{proof}
\noindent\emph{Correctness}:
By the input requirement, $S$ is not a witness.
Hence, either a larger predecessor is found or the case of nonempty complement intersection from line~9 occurs for some of the predecessors of $S$.
In the latter case, the function continues from the larger union, which is also reachable.
Hence, it follows by induction on the steps with the latter case that finally a larger predecessor must be found in line~19.
Thus, queue $\mathit{Process}$ is never emptied and the function ends with the call to \Call{Reconstruct}{}.

For the reconstruction phase, we show that, except for the beginning of the first iteration, the following invariants are kept between iterations:
\begin{enumerate}
\item $w$ is a properly extending word for the current set $C$, and
\item $\mathit{Trace}[c]$ is not \none.
\end{enumerate}
We show this by induction on iterations of the \textbf{while} loop in line~24.
In the first iteration, $w$ is empty and $C$ is the given set $T'$ from the call in line~20.
Thus $\mathit{Trace}[C]$ is a letter and $|C| > |\delta(C,\mathit{Trace}[C])| = |T|$, so we set $w = a$ and the invariant holds for the next iteration.
In the other iterations, by the inductive assumption, $w$ is a properly extending word for the current $C$ at the beginning of the iteration.
$\mathit{Trace}[C]$ is not \none, and we have two cases depending on whether $\mathit{Trace}[C]$ is a letter.
If this is a letter $a$, then the invariant is kept by the next application of $a$ (lines~27--29) and setting $C = \delta(C,a)$.
Because the letter $a$ has been put to $\mathit{Trace}$ in line~18, where $T'$ is $a$-predecessor of $T$, the set $T = \delta(T',a)$ has been processed previously so $\mathit{Trace}[T']$ is not \none, and our new $C$ is $T'$.
In the second case, we recall the condition $T \subsetneq T \cup T' = C \subsetneq Q$ from line~9, which implies that a properly extending word for $C$ is properly extending for either $T$ or $T'$ by~\Cref{lem:UnionPredecessor}.
We directly check which case holds in line~32, hence $w$ is a properly extending word for the newly chosen $C$.

Finally, we observe that the reconstruction terminates with $C = S$.
In both cases depending on whether $\mathit{Trace}[C]$ is a letter, the next set $C$ is also present in $\mathit{Trace}$ and it has been added in an earlier iteration of the while loop in line~6, or it is $S$ added in line~4.
Hence, finally, it must be the latter case.

Concluding, in the end, we have $C = S$, so both functions return a word $w$ that is properly extending for $S$.
\end{proof}

The remaining effort is to prove an upper bound on the length of the returned word by~\Call{FindShortProperlyExtendingWord}{}.

\subsubsection{Trivally intersecting subsets -- a combinatorial problem}

To provide a bound on the length of the found word, we consider the following auxiliary combinatorial problem.
Let $Q$ an $n$-element universe.
Two subsets $S, T \subseteq Q$ are \emph{trivially intersecting} if $S \cap T \in \{\emptyset,S,T\}$, i.e., either they are disjoint or one set contains the another.
A family of subsets of $Q$ is called \emph{trivially intersecting} if all the subsets are pairwise trivially intersecting.

\begin{definition}
For an integer $n \ge 1$, the number \emph{$\mathit{TriviallyIntersectingSubsets}(n)$} is the maximum size of a trivially intersecting family for an $n$-element universe.
\end{definition}

\begin{lemma}\label{lem:TriviallyIntersectingSubsets}
$\mathit{TriviallyIntersectingSubsets}(n) = 2n$.
\end{lemma}
\begin{proof}
It goes by induction on $n$.
For the base case, $\mathit{TriviallyIntersectingSubsets}(1) = 2$ (the $1$-element set and the empty set).

For $n > 1$, consider a largest trivially intersecting family.
Besides the set with all elements and the empty set, there should be exactly two proper subsets of $Q$ that are maximal with respect to inclusion.
Otherwise, if there are more, we could add a union of two of them, and if there is only one, we could add its complement, in both cases obtaining a larger trivially intersecting family.
Thus, the solution to the problem is recursive.
We use the inductive assumption two times excluding the empty set, which is counted separately.
\begin{align*}
\mathit{TriviallyIntersectingSubsets}(n) =&\ 2 + \max_{p \in \{1,\ldots,n-1\}} (\mathit{TriviallyIntersectingSubsets}(p) + \mathit{TriviallyIntersectingSubsets}(n-p)) \\
=&\ 2 + \max_{p \in \{1,\ldots,n-1\}} (2p-1 + 2(n-p)-1) \\
=&\ 2 + \max_{p \in \{1,\ldots,n-1\}} (2n-2) \\
=&\ 2 + 2n-2 = 2n.
\end{align*}
\end{proof}

Now, we consider a generalized version of the problem that is needed for our analysis of the algorithm.
We add limits on the maximum and minimum size of the sets in the family.

\begin{definition}
For integers $n \ge 1$, $0 \le \ell,k \le n$, the number $\mathit{TriviallyIntersectingSubsets}^{\le k}_{\ge \ell}(n)$ is the maximum size of a trivially intersecting family for an $n$-element universe where each subset from the family has size at most $k$ and at least $\ell$.
When $\ell = 0$ or $k = n$, we omit the subscript or the superscript, respectively.
\end{definition}

\begin{lemma}\label{lem:BoundedTriviallyIntersectingSubsets}
$\mathit{TriviallyIntersectingSubsets}^{\le k}(n) = 2n+1 - \lceil n/k \rceil$.
\end{lemma}
\begin{proof}
We prove the equality by induction descending on $k$ and ascending on $n$.
The base case $k = n$ is solved by~\Cref{lem:TriviallyIntersectingSubsets}, which also is the only possible case for $n=1$.

Let $1 \le k < n$ and assume that the equality holds for all $\mathit{TriviallyIntersectingSubsets}^{\le k'}(n')$, where $k' > k$ and $n'=n$, or $k' = k$ and $n' < n$.
Consider a trivially intersecting family for $k$ and $n$ that has the maximum size $\mathit{TriviallyIntersectingSubsets}^{\le k}(n)$.

If the largest subsets in the family have size $k$, then consider separately subsets of one of these largest sets and its complement, thus we split the problem into two, so our family has size:
\[ \mathit{TriviallyIntersectingSubsets}(k) + \mathit{TriviallyIntersectingSubsets}^{\le k}(n-k) - 1, \]
where $-1$ at the end comes from not counting the empty set twice.
By~\Cref{lem:TriviallyIntersectingSubsets} for $\mathit{TriviallyIntersectingSubsets}(k)$ and the inductive assumption for $\mathit{TriviallyIntersectingSubsets}^{\le k}(n-k)$:
\[ \mathit{TriviallyIntersectingSubsets}(k) + \mathit{TriviallyIntersectingSubsets}^{\le k}(n-k) - 1 = 2k + 2(n-k)+1 - \lceil (n-k)/k \rceil - 1 = 2n+1 - \lceil n/k \rceil. \]

If the largest subsets have size $k' < k$, then from the inductive assumption the size of the family equals
\[ \mathit{TriviallyIntersectingSubsets}^{\le k'}(n) = 2n+1 - \lceil n/k' \rceil \le 2n+1 - \lceil n/k \rceil. \]
Thus, the number is not larger than in the previous case, and by taking the maximum from the two possibilities we get the claim.
\end{proof}

Finally, in the particular variant that we need, we exclude singletons and the empty set.

\begin{corollary}
$\mathit{TriviallyIntersectingSubsets}^{\le k}_{\ge 2}(n) = n - \lceil n/k \rceil$.
\end{corollary}
\begin{proof}
$n$ singletons and the empty set are present in every maximal trivially intersecting family.
\end{proof}

\subsubsection{Bounding reaching thresholds}

We apply the above combinatorial problem to derive an upper bound on the length of a word found by \Call{FindShortProperlyExtendingWord}{}.
The crucial property is that the family of the complements of all processed subsets $T$ for which the condition in line~9 did not hold is trivially intersecting.
The upper bound follows since all the letters in the final reconstructed word were added in the second case, in line~18, and such a letter can be added only once to the final word.
Therefore, each such letter implies one more set in our trivially intersecting family.
Additionally, we count sets of size $n-1$ separately (which are singletons in the family of the complements), since their number is upper bounded by the maximum size of a group orbit.

\begin{lemma}\label{lem:FindShortProperlyExtendingWord}
Let $\ell$ be the maximum size of a group orbit of the automaton.
For a nonempty proper subset $S \subsetneq Q$, the word returned by~\Call{FindShortProperlyExtendingWord}{} from~\Cref{alg:FindShortProperlyExtendingWord} has length at most $\ell + n - \lceil n/(n-|S|)\rceil$.
\end{lemma}
\begin{proof}
Let $\mathcal{X}$ be the family of all processed subsets $T$ for which the condition in line~9 did not hold, i.e., the sets $T$ in the iterations that passed to line~14.
Note that $S \in \mathcal{X}$ and $Q \notin \mathcal{X}$ (because obtaining $Q$ is possible only in direct extending in line~16 and always ends the search).

Observe that the size of $\mathcal{X}$ is an upper bound on the length of the reconstructed word:
If for some set $C$, $\mathit{Trace}[C]$ is a letter, then $\delta(C,\mathit{Trace}[C]) \in \mathcal{X}$.
Since during the reconstruction, the sets $C$ do not repeat, each time a letter is used in line~28, a new set $\delta(C,\mathit{Trace}[C])$ is added to $\mathcal{X}$.

Let $\overline{\mathcal{X}}$ denote $\mathcal{X}$ where each subset is replaced with its complement.
Now we show that the family $\overline{\mathcal{X}}$ is trivially intersecting.
All sets from $\mathcal{X}$ are also present as keys in $\mathit{Trace}$ (however, $\mathit{Trace}$ can contain more sets, which were either added in line~10 or line~18 and not yet processed).
Let $T \in \mathcal{X}$ be one of the sets.
So the condition from line~9 did not hold together with each $T' \in \mathcal{X}$ inserted before $T$.
Note that the sets $T$ taken from queue $\mathit{Proces}$ are processed in nondecreasing order by size, thus it cannot happen that $T \subsetneq T'$.
Therefore, the condition $T \subsetneq T \cup T' \subseteq Q$ is equivalent to that $T,T' \subsetneq T \cup T' \subsetneq Q$.
Hence it is symmetric, so it fails for every pair of states $T,T' \in \mathcal{X}$.
The negation of this condition implies that we have $T \subseteq T'$, $T' \subseteq T$, or $T \cup T' = Q$, which is equivalent to that $\overline{T} \cap \overline{T'} \in \{\emptyset,\overline{T},\overline{T'}\}$.
Thus, the family $\overline{\mathcal{X}}$ is trivially intersecting.

The sets $T \in \mathcal{X}$ have size $\ge |S|$, so their complements have size $\le n-|S|$, and $\mathit{TriviallyIntersectingSubsets}^{\le n-|S|}(n)$ is an upper bound on $|\overline{\mathcal{X}}| = |\mathcal{X}|$.
Also, the number of sets $T \in \mathcal{X}$ of size $n-1$ is at most $\ell$ since these sets cannot be reached from any set of size $n-1$ that misses its state from a different group orbit than the missing state of $T$.
Hence, we can bound their number by $\ell$, and bound the number of the rest of the sets $\mathcal{X} \setminus \{Q \setminus \{q\} \mid q \in Q\}$ by $\mathit{TriviallyIntersectingSubsets}^{\le n-|S|}_{\ge 2}(n)$.
Finally, we have:
\[ |\mathcal{X}| \le \ell + \mathit{TriviallyIntersectingSubsets}^{\le n-|S|}_{\ge 2}(n) = \ell + n - \lceil n/(n-|S|)\rceil,\]
which is our upper bound on the length of the found $w$.
\end{proof}

Finally, using the standard extension method (starting from the given subset $S$ and iteratively extending it to $Q$) and some calculations, we obtain an upper bound on the reaching threshold:

\begin{theorem}\label{thm:ReachingThresholdBound}
Let $(Q,\Sigma,\delta)$ be an automaton with the maximum size $\ell$ of a group orbit, let $S \subsetneq Q$ be a nonempty subset that is reachable, and assume that all subsets of $Q$ of size $> |S|$ are reachable.
Then $S$ is reachable with a word of length at most
\[ (n-|S|)(\ell+n) - n \cdot H_{n-|S|} ,\]
where $H_i$ is the $i$-th harmonic number.
\end{theorem}
\begin{proof}
Starting from $S$, we iteratively apply at most $n-|S|$ times a properly extending word returned by~\Call{FindShortProperlyExtendingWord}{}.
By~\Cref{lem:FindShortProperlyExtendingWord}, the sum of the lengths of the obtained words is at most
\begin{align*}
\sum_{k=|S|}^{n-1} \left( \ell + n - \lceil n/(n-k)\rceil \right) &\ \le (n-|S|)(\ell+n) - n\cdot \sum_{k=|S|}^{n-1} 1/(n-k) \\
&\ = (n-|S|)(\ell+n) - n \cdot \sum_{k=1}^{n-|S|} 1/k \\
&\ = (n-|S|)(\ell+n) - n \cdot H_{n-|S|} .
\end{align*}
\end{proof}

%%%%%%%%%%%%%%%%%%%%%%%%%%%%%%%%%%%%%%%%%%%%%%%%%%%%%%%%%%%%
\section{Conclusions and discussion}\label{sec:Conclusions}

In the first part of the paper, we discussed witnesses as a tool for determining the complete reachability.
We have shown a polynomial-time algorithm for finding a maximal witness -- a largest unreachable subset of states.

\begin{corollary}
Problem~\textsc{CompletelyReachable} can be solved in $\O(|\Sigma|\cdot n^2)$ time and in $\O(|\Sigma|\cdot n)$ space.
\end{corollary}
It can be noted that the complexity of our algorithm is the same time complexity as that of the fastest known algorithm deciding the synchronizability of an automaton \cite{Eppstein1990,Volkov2022Survey}, although these algorithms are entirely different.

There is another interesting side corollary from our algorithm, concerning the maximal size of a minimal alphabet required to make an automaton completely reachable.
The alphabet of a completely reachable automaton is \emph{minimal with respect to complete reachability} if removing any of the letters makes it not completely reachable.
It turns out that such a minimal alphabet has at most $2n-2$ letters.
Furthermore, this bound is tight.

\begin{corollary}
If an $n$-state completely reachable automaton has a minimal alphabet with respect to complete reachability, then the size of this alphabet is at most $2n-2$.
\end{corollary}
\begin{proof}
Let $\Sigma$ be the alphabet of a completely reachable automaton.
Recall that \Call{FindReduction}{} from~\Cref{alg:FindReduction} finds just one letter to derive a reduction, and the existence of this letter alone is sufficient for this call of the function.
Suppose that we run the whole algorithm \Call{FindMaximalWitness}{} from~\Cref{alg:FindMaximalWitness} 
and construct $\Sigma' \subseteq \Sigma$ containing only these letters that have been used for deriving reductions in~\Call{FindReduction}{}.
We replace $\Sigma$ with $\Sigma'$, and run the algorithm again for the obtained automaton.
Then, the execution of the algorithm will be the same, with the possible exception of more iterations in~\Call{FindReduction}{}, which do not change the outcome of that function.
Hence, the computed reduction graphs will be the same and the algorithm outputs that the automaton is completely reachable.

The total number of reductions found in~\Call{FindReduction}{} equals the number of updates of the reduction graph, which is limited by $2n-2$ \Cref{lem:UpdatesReductionGraph}.
Thus, if $|\Sigma| > 2n-2$, the alphabet is not minimal with respect to complete reachability.
\end{proof}

\begin{proposition}
For every $n \ge 1$, there exists a completely reachable automaton with a minimal alphabet with respect to complete reachability of size $2n-2$.
\end{proposition}
\begin{proof}
We define an automaton where states are linearly ordered, and for each adjacent pair of states in this order, there are two letters with the actions mapping one state to the other and vice versa, while the other states are fixed.
Let $\mathrsfs{A} = (Q,\Sigma,\delta)$, where $Q = \{q_0,\ldots,q_{n-1}\}$ and $\Sigma = \{a_0,\ldots,a_{n-2},b_0,\ldots,b_{n-2}\}$.
For each $0 \le i < n$, we define $\delta(q_i,a_i)=q_{i+1}$ and $\delta(q_{i+1},b_i)=q_i$, and we let the other states be fixed by the action of these letters.

The automaton is completely reachable:
For a nonempty subset $S \subsetneq Q$, we can find a pair of states $q_i,q_{i+1}$ such that either $q_i \in S$ and $q_{i+1} \notin S$ or vice versa.
Then either $b_i$ or $a_i$, respectively, is properly extending for $S$.
It follows by~\Cref{rem:Characterization} that the automaton is completely reachable.

The alphabet is minimal with respect to complete reachability:
Removing any letter yields an automaton that is not strongly connected, which implies that it cannot be completely reachable \cite{BV2016CompletelyReachableAutomata}.
\end{proof}

In the second part of the paper, we have proved a quadratic upper bound on the length of the shortest words reaching a given subset (the reaching threshold of the subset).
It has been obtained with a different algorithm using the complement-intersection technique.

From~\Cref{thm:ReachingThresholdBound} with the maximum possible size of a group orbit bounded trivially by $n$, we get Don's conjecture (\cite{D2016OneContracting}, \cite[Problem~4]{GJ2019HardlyReachableSubsets}) weakened by the factor of $2$.
\begin{corollary}\label{cor:WeakDon}
In an $n$-state completely reachable automaton $(Q,\Sigma,\delta)$, every nonempty subset $S \subsetneq Q$ is reachable with a word of length smaller than $2n(n-|S|)$.
\end{corollary}

However, the actual upper bound is slightly smaller, which especially matters in certain cases, e.g., when $|S|$ is large.
\Cref{tab:ReachabilityLargeSize} shows the upper bound derived for subsets of size close to $n$.
The bound for $|S|=n-1$ is trivially tight, whereas for the others, it is an open problem.

\begin{table}[htb!]\centering\renewcommand{\arraystretch}{1.2}
\caption{The upper bound from~\Cref{thm:ReachingThresholdBound} compared to the simplified one from~\Cref{cor:WeakDon} on the reaching threshold of a subset of a large size.}\label{tab:ReachabilityLargeSize}
$\begin{array}{|c|r|r|r|r|r|r|r|r|}\hline
|S|                           & n-1 & n-2 & n-3 & n-4 & n-5 & n-6 & n-7 & n-8 \\\hline
2n(n-|S|) & 2\ n & 4\ n & 6\ n & 8\ n & 10\ n & 12\ n & 14\ n & 16\ n \\[4pt]
\multirow{2}{*}{$2n(n-|S|) - n \cdot H_{n-|S|}$} & n & 5/2\ n & 25/6\ n & 71/12\ n & 463/60\ n & 191/20\ n & 1597/140\ n & 3719/280\ n \\
& = n & = 2.5\ n & \simeq 4.17\ n & \simeq 5.92\ n & \simeq 7.72\ n & = 9.55\ n & \simeq 11.41\ n & \simeq\ 13.28 n \\\hline
\end{array}$\end{table}

Recently, the strict Don's conjecture has been disproved for the binary case of a completely reachable automaton and $|S|=n-2$ \cite{Zhu2024AroundDonsConjecture}, by exhibiting an infinite series of automata with subsets whose reaching threshold is $5/2n-3$.
Hence, for this case, we know an upper bound that is almost tight.
We note that, in the binary case, we always have $\ell = n$, as one of the letters must induce a full cycle on $Q$.

There are also better upper bounds than ours provided for certain subfamilies of binary completely reachable automata \cite{CV2023DonsConjectureForBinaryCompletelyReachable,Zhu2024AroundDonsConjecture}.
These considerations, however, do not apply to the general case, as they rely on the specific structure that a binary completely reachable automaton must have.

Using a lower bound on the harmonic numbers, we can also estimate our upper bound on reaching thresholds from~\Cref{thm:ReachingThresholdBound} with a more closed formula.
\begin{corollary}\label{cor:ReachingThresholdEstimation}
The upper bound on the reaching threshold from~\Cref{thm:ReachingThresholdBound} can be estimated as follows:
\[ (n-|S|)(\ell+n) - n \cdot H_{n-|S|} < (n-|S|)(n+\ell) - n \ln (n-|S|) - \gamma n ,\]
where $\gamma \simeq 0.5772156649$ is the Euler-Mascheroni constant.
\end{corollary}
\begin{proof}
We apply the inequality \cite{ChenQi2008TheBestBoundsForHarmonicHumber} $\ln(i) + \gamma + 1/(2i+1) \le H_i$, which holds for each natural number $i$.
\end{proof}

We can also get strict Don's conjecture when the group orbits are small, i.e., logarithmic in $n$.
This, for example, applies to the case where a completely reachable automaton does not have any permutational letters.
\begin{corollary}\label{cor:DonForLogEll}
For an $n \ge 3$, in an $n$-state completely reachable automaton $(Q,\Sigma,\delta)$ where the maximum size of a group orbit is $\ell \le \ln n$, every nonempty subset $S \subsetneq Q$ is reachable with a word of length at most $n(n-|S|)$.
\end{corollary}
\begin{proof}
We take the bound from~\Cref{cor:ReachingThresholdEstimation} and apply $\ell \le \ln n$, thus we have the upper bound:
\[ (n-|S|) (n+\ln n) - n \ln (n-|S|) - \gamma n = n(n-|S|) + (n-|S|)\ln n - n \ln (n-|S|) - \gamma n .\]
We show that $(n-|S|)\ln n - n \ln (n-|S|) - \gamma n \le 0$, which then implies the upper bound $n(n-|S|)$.

Let $x = n-|S|$.
We observe that the inequality $x\ln n \le n \ln x + \gamma n$ holds for all $1 \le x \le n$, assuming $n \ge 3$.
It is easy to check this for $x \in \{1,2\}$ with $n \ge 1$.
For $x \ge 3$, we have the stronger inequality $\ln n / n < \ln x / x$, which holds because the function $\ln i / i$ for $i \ge 3$ is decreasing and we have $n > x$.
This is equivalent to $x \ln n < n \ln x$, which implies our weaker inequality.
\end{proof}

Considering the \Cerny problem, we get a quadratic upper bound for completely reachable automata.
This additionally slightly improves over the bounds for previously studied subclasses (mentioned in the introduction) of completely reachable automata, which were of order $2n - \O(n)$.

\begin{corollary}
The reset threshold of a completely reachable automaton with $n \ge 2$ states is at most
\[ 1 + 2n(n-2) - n \cdot H_{n-2} \le 2n(n-2) - n \ln (n-2) - \gamma n + 1,\]
where $\gamma \simeq 0.5772156649$ is the Euler-Mascheroni constant.
The bound can be simplified to $2n(n-2) - n \ln n$ for $n \ge 6$.

When the maximum size of a group orbit is $\ell$, the upper bound can be improved to:
\[ 1 + (n-2)(n+\ell) - n \cdot H_{n-2} \le (n-2)(n+\ell) - n \ln (n-2) - \gamma n + 1 .\]
\end{corollary}
\begin{proof}
We find a state $q$ such that $|\delta^{-1}(\{q\},a)| \ge 2$ for some letter $a$.
Then we (properly) extend $\delta^{-1}(\{q\},a)$ of size at least $2$ by the word from~\Cref{thm:ReachingThresholdBound}.
Together with the first letter, this gives the upper bound:
\[ 1 +  (n-2)(\ell+n) - n \cdot H_{n-2} .\]
Depending on $\ell$, the bound takes its maximum for $\ell = n$, where we get $1 + 2n(n-2) - n \cdot H_{n-2}$.

By~\Cref{cor:ReachingThresholdEstimation}, we also get the claimed inequalities.
Finally, since $\ln(n-2) + \gamma - 1/n \ge \ln n$ for $n \ge 6$, we get the simplification.
\end{proof}

We obtain the \Cerny conjecture for the cases of completely reachable automata with the maximum size of a group orbit $\ell \le \ln n$.
This follows from~\Cref{cor:DonForLogEll} for a subset of size at least $2$ that is the preimage of a singleton under the inverse action of one letter: we get $1+n(n-2) = (n-1)^2$.

Finally, we note that relaxing the complete reachability to the reachability of only subsets that are large enough would be sufficient to derive a subcubic upper bound on the reset threshold.
It follows by the generalized method of \emph{avoiding words} \cite[Theorem~12]{FSV21LowerBoundsOnAvoidingThresholds} by the fact that reaching a set means avoiding its complement.
A word $w$ is \emph{avoiding} for a subset $S \subseteq Q$ if $S \cap \delta(Q,w) = \emptyset$.

For the statement, we recall the well-known notation that $\omega(1)$ is the set of functions growing asymptotically faster than every constant function (i.e., if $f \in \omega(1)$, then for every constant $c$, we have $f(n) > c$ for a large enough $n$), and $o(n^3)$ is the set of functions growing asymptotically slower than $n^3$.
\begin{corollary}
The reset threshold of an $n$-state automaton $(Q,\Sigma,\delta)$ where all subsets of size at least $n - \omega(1)$ are reachable, where $\omega(1)$ is the set of functions growing faster than a constant, is in $o(n^3)$. 
\end{corollary}
\begin{proof}
Suppose that for a function $f(n) \in \omega(1)$, all subsets $S$ of size $\ge n - f(n)$ are reachable.
Then by~\Cref{thm:ReachingThresholdBound}, they are reachable with a word of length at most $2n(n-|S|) \in \O(n\cdot f(n))$.
Thus, in particular, all subsets of size at most $f(n)$ are avoidable by reaching their complement with a word of length at most $\O(n\cdot f(n)))$.
By~\cite[Theorem~12]{FSV21LowerBoundsOnAvoidingThresholds}, the reset threshold of the automaton is in $o(n^3)$.
\end{proof}

%%%%%%%%%%%%%%%%%%%%%%%%%%%%%%%%%%%%%%%%%%%%%%%%%%%%%%%%%%%%
\section*{Acknowledgments}
This work was supported in part by the National Science Centre, Poland under project number 2017/26/E/ST6/00191 (Robert Ferens) and 2021/41/B/ST6/03691 (Marek Szyku{\l}a).

%%%%%%%%%%%%%%%%%%%%%%%%%%%%%%%%%%%%%%%%%%%%%%%%%%%%%%%%%%%%
\bibliographystyle{plain}
\bibliography{synchronization}
%%%%%%%%%%%%%%%%%%%%%%%%%%%%%%%%%%%%%%%%%%%%%%%%%%%%%%%%%%%%
\section*{Appendix}

\begin{theorem}\label{thm:PEW-PSPACE}
For a given automaton $(Q,\Sigma,\delta)$ and a set $S \subseteq Q$, verifying whether there exists a properly extending word for $S$ is PSPACE-complete.
\end{theorem}
\begin{proof}
Solving the problem in NPSPACE, thus in PSPACE, is straightforward.
We start from $S$ and guess letters one by one, applying their inverse action to the current subset.
A letter $a$ can be applied to the current set $S'$ only if $\delta^{-1}(S',a)$ is an $a$-predecessor -- every state $q \in T$ has the nonempty preimage $\delta^{-1}(q,a)$.
We accept if we encounter a set larger than $|S|$ and reject if we reach the limit of $2^n$ letters.

\noindent\textbf{Reduced problem:}
For PSPACE-hardness, we reduce from the well-known problem called \textsc{Finite Automata Intersection}.

To clarify the forthcoming notation, we precise that by the superscript we denote the automaton to which an object (usually a state, a letter, or a subset) belongs, and the subscript denotes the index (usually of a state or a letter).
The superscripts and subscripts are always natural numbers.

We are given $m \ge 1$ automata $\mathrsfs{A}_i = (Q_i,\Sigma,\delta_i)$ with a common alphabet $\Sigma$, pairwise disjoint sets of states $Q_i$, initial states $\mathit{init}_i \in Q_i$, and sets of final states $F_i \subseteq Q_i$, where $i \in \{1,\ldots,m\}$.
The decision problem is whether there exists a word $w \in \Sigma^*$ such that $\delta_i(\mathit{init}_i,w) \in F_i$ for all $i$.
Such words are called \emph{accepted}.

It is known that this problem is PSPACE-complete \cite{K1977LowerBoundsForNaturalProofSystems}, and since in the original construction the given automata have one unique final state, we can further assume that here.
Hence, denote for each $i$, $F_i = \{\mathit{final}_i\}$.
Also, denote the set of all initial states $I = \{\mathit{init}_i \mid 1 \le i \le m\}$ and the set of all final states $F = \{\mathit{final}_i \mid 1 \le i \le m\}$.

\noindent\textbf{Idea:}
The idea of the reduction is as follows.
We can ask whether, for the set of all final states $F$, we can obtain a preimage containing the initial states $I$ -- this is equivalent to the existence of an accepted word.
Then we can allow extending $I$ to a larger set.
However, the other preimages of $F$ can also be larger than $m$, and some additional states in intermediate subsets may have an empty preimage.
Hence, we need to modify the construction so that only subsets of size $m$ are possible (until $I$ is reached) with exactly one state in each $Q_i$; otherwise, an obtained preimage is not a predecessor.
Finally, for $I$, we add a special letter that will allow, exclusively for this set, to obtain a larger predecessor of size $m+1$.

We introduce \emph{transitional states} and \emph{transitional letters}, which allow choosing one particular transition from a few transitions labeled by the same letter and going to the same state.
Then the states in $Q$ coming from the given automata are called \emph{main states}, and the letters in $\Sigma$ are called \emph{main letters}.
For instance, for an $i$-th automaton, if we have two incoming transitions $x \xrightarrow{a} z$ and $y \xrightarrow{a} z$ to the same state that are labeled by the same letter, where $x,y,z \in Q_i$ and $a \in \Sigma$, then we introduce transitional states $t_{a,x}$, $t_{a,y}$ and transitional letters $c_{a,x}$, $c_{a,y}$.
The action of these letters will map the corresponding transitional state to $z$.
Note that it is enough to parameterize transitional states and letters with the chosen letter and the state chosen to be in the preimage ($x$ or $y$); the destination state $z = \delta_i(x,a) = \delta_i(y,a)$ is uniquely defined by them as our automata are deterministic.
The described scheme is illustrated in~\Cref{fig:PEW-Scheme}.

Once we have chosen one transition independently (labeled by the same letter, however) in every automaton, the inverse action of main letters will serve to obtain the main states from the corresponding transitional states.

\begin{figure}\centering\large\begin{tikzpicture}[node distance=.5cm,scale=1,every node/.style={transform shape}]
\tikzset{every loop/.style={min distance=.5cm,looseness=6}}
\node[state, inner sep=0pt] (qz) {$z$};
\node[state, inner sep=0pt] [above left=1.37cm and 2cm of qz] (qx) {$x$};
\node[state, inner sep=0pt] [below left=1.37cm and 2cm of qz] (qy) {$y$};
\draw (qx) edge[] node[auto,midway]{$a$} (qz);
\draw (qy) edge[] node[auto,midway,swap]{$a$} (qz);
\end{tikzpicture}\hspace{2cm}%
\begin{tikzpicture}[node distance=2cm,scale=1,every node/.style={transform shape}]
\node[state, inner sep=0pt] (qz) {$z$};
\node[state, inner sep=0pt] [above left=.4cm and 1cm of qz] (tx) {$t_{a,x}$};
\node[state, inner sep=0pt] [below left=.4cm and 1cm of qz] (ty) {$t_{a,y}$};
\node[state, inner sep=0pt] [above left=.4cm and 1cm of tx] (qx) {$x$};
\node[state, inner sep=0pt] [below left=.4cm and 1cm of ty] (qy) {$y$};
\draw (qx) edge[] node[auto,midway]{$a$} (tx);
\draw (tx) edge[] node[auto,midway]{$c_{a,x}$} (qz);
\draw (qy) edge[] node[auto,midway,swap]{$a$} (ty);
\draw (ty) edge[] node[auto,midway,swap]{$c_{a,y}$} (qz);
\end{tikzpicture}
\caption{The scheme of the reduction construction in the proof of~\Cref{thm:PEW-PSPACE}.}\label{fig:PEW-Scheme}
\end{figure}
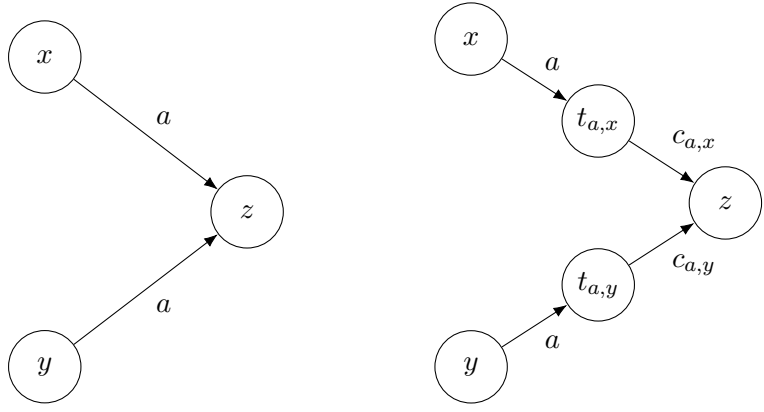

There is a special state $\mathit{trash}$ for utilizing superfluous transitions.
The transitional letters act as the identity on the main states of the other automata but send to $\mathit{trash}$ all main states within the same automaton.
Hence, it will not be possible to use a transitional letter created for one automaton more than once within one substring not containing main letters, as there will not be an incoming transition to the currently active transitional states.
Additionally, they act as the identity on the transitional states of the same letter in the other automata but also send to $\mathit{trash}$ those that do not agree on the chosen letter.
This will enforce that once a main letter is chosen, only transitional letters that agree with this choice can be used.
Hence, for each automaton, exactly once a transitional letter of this automaton must be used, as main letters will not have incoming transitions to main states.

Finally, we need a special letter $\mathit{done}$, which will be allowed only when the set $I$ is reached.
It will then act to extend $I$ to a larger set.

\noindent\textbf{Formal construction:}
We build $\mathrsfs{A} = (Q,\Gamma,\delta)$, where:
\begin{itemize}
\item $Q = \big(\bigcup_{1 \le i \le m} Q_i \cup T_i\big) \ \cup\ \{\mathit{trash}\}$, where: \begin{itemize}
\item For each $i$, $T_i = \{t_{\alpha,x} \mid \alpha \in \Sigma, x \in Q_i\}$ is the set of \emph{transitional states} for $\mathrsfs{A}_i$.
\item $\mathit{trash}$ is a unique fresh state \hfill (for trashing superfluous transitions).
\end{itemize}
\item $\Gamma = \Sigma\ \cup\ \big(\bigcup_{1 \le i \le m} \Psi_i\big)\ \cup\ \{\mathit{done}\}$, where:
\begin{itemize}
\item For each $i$, $\Psi_i = \{c_{\alpha,x} \mid \alpha \in \Sigma, x \in Q_i\}$ is the set of \emph{transitional letters} for $\mathrsfs{A}_i$.
\item $\mathit{done}$ is a unique fresh letter \hfill (for obtaining finally a larger preimage from $I$).
\end{itemize}
\item $\delta$ is defined as follows:
\begin{itemize}
\item For each $1 \le i \le m$, for each $\alpha \in \Sigma$, and for each $x \in Q_i$:\\
\begin{itemize}
\item $\delta(x,\alpha) = t_{\alpha,x}$ and $\delta(t_{\alpha,x},c_{\alpha,x}) = \delta_i(x,\alpha)$ \hfill (see~\Cref{fig:PEW-Scheme}).
\item For each $j \neq i$, for each $y \in Q_j$: \hfill (for each state of the other automata)\\
$\delta(t_{\alpha,y},c_{\alpha,x}) = t_{\alpha,y}$ \hfill (all their transitional states of the same letter are fixed);\\
$\delta(y,c_{\alpha,x}) = y$ \hfill (all their main states are fixed).
\end{itemize}
\item For each $1 \le i \le m$, $\delta(\mathit{init}_i,\mathit{done}) = \mathit{init}_i$ \hfill (initial states are fixed).
\item $\delta(\mathit{trash},\mathit{done}) = \mathit{init}_1$ \hfill (for getting an extra state in the final larger preimage).
\item All remaining transitions not defined above go to $\mathit{trash}$.
\end{itemize}
\end{itemize}

We ask whether in $\mathrsfs{A}$, there exists a properly extending word for the set of final states $F$.

\Cref{fig:ExamplePEWReduction} shows the construction for a problem instance with two given automata.

\begin{figure}[htb!]\large\begin{tikzpicture}[node distance=3cm,scale=.9,every node/.style={transform shape},bend angle=20]
%\tikzset{every loop/.style={min distance=2cm,looseness=6}}
\node[state,initial,initial text=\phantom{aa}] (q1){$q_1$};
\node[state,accepting] [right=3cm of q1](q2){$q_2$};
%\node[] [above=2cm of q1](labelA1){$\mathrsfs{A}_1$:}

\draw(q1) edge[loop,out=120,in=60,min distance=1cm,looseness=6] node[auto,midway]{$a$} (q1);
\draw(q1) edge[bend left] node[auto,midway]{$b$} (q2);
\draw(q2) edge[bend left] node[auto,midway]{$a$} (q1);
\draw(q2) edge[loop,out=120,in=60,min distance=1cm,looseness=6] node[auto,midway]{$b$} (q2);

\node[state,initial,initial text=,accepting] [right=2cm of q2](p1) {$p_1$};
\node[state] [right=3.15cm of p1](p2){$p_2$};
\node[state] [right=3.2cm of p2](p3){$p_3$};
\draw(p1) edge node[auto,midway]{$b$} (p2);
\draw(p2) edge node[auto,midway]{$a$} (p3);
\draw(p3) edge[bend angle=30,bend left] node[auto,midway]{$a,b$} (p1);
\draw(p1) edge[loop,out=120,in=60,min distance=1cm,looseness=6] node[auto,midway]{$a$} (p1);
\draw(p2) edge[loop,out=120,in=60,min distance=1cm,looseness=6] node[auto,midway]{$b$} (p2);
\end{tikzpicture}\\\\
\begin{tikzpicture}[node distance=1cm,scale=.9,every node/.style={transform shape},bend angle=20]
\node[state] (q1) {$q_1$};
\node[state] [above right=.2cm and 1.2cm of q1](tbq1) {$t_{b,q_1}$};
\node[state] [below right=.2cm and 1.2cm of tbq1](q2) {$q_2$};
\node[state] [below right=.2cm and 1.2cm of q1](taq2) {$t_{a,q_2}$};
\node[state] [above=1.5cm of q1](taq1) {$t_{a,q_1}$};
\node[state] [above=1.5cm of q2](tbq2) {$t_{b,q_2}$};

\draw(q1) edge[bend left] node[auto,midway]{$a$} (taq1);
\draw(taq1) edge[bend left] node[auto,midway]{$c_{a,q_1}$} (q1);
\draw(q2) edge[bend left] node[auto,midway]{$b$} (tbq2);
\draw(tbq2) edge[bend left] node[auto,midway]{$c_{b,q_2}$} (q2);

\draw(q1) edge[bend angle=10,bend left] node[auto,midway,swap]{$b$} (tbq1);
\draw(tbq1) edge[bend angle=10,bend left] node[auto,midway,swap,pos=0.7]{$c_{b,q_1}$} (q2);
\draw(q2) edge[bend angle=10,bend left] node[auto,midway]{$a$} (taq2);
\draw(taq2) edge[bend angle=10,bend left] node[auto,midway,pos=.2]{$c_{a,q_2}$} (q1);

\node[state] [right=2cm of q2](p1){$p_1$};
\node[state] [right=1cm of p1](tbp1){$t_{b,p_1}$};
\node[state] [right=1cm of tbp1](p2){$p_2$};
\node[state] [right=1cm of p2](tap2){$t_{a,p_2}$};
\node[state] [right=1cm of tap2](p3){$p_3$};
\node[state] [below=1.0cm of p2](tap3){$t_{a,p_3}$};
\node[state] [below=.5cm of tap3](tbp3){$t_{b,p_3}$};
\node[state] [above=1.5cm of p1](tap1){$t_{a,p_1}$};
\node[state] [above=1.5cm of p2](tbp2){$t_{b,p_2}$};

\draw(p1) edge[bend left] node[auto,midway]{$a$} (tap1);
\draw(tap1) edge[bend left] node[auto,midway]{$c_{a,p_1}$} (p1);
\draw(p2) edge[bend left] node[auto,midway]{$b$} (tbp2);
\draw(tbp2) edge[bend left] node[auto,midway]{$c_{b,p_2}$} (p2);

\draw(p1) edge[] node[auto,midway]{$b$} (tbp1);
\draw(tbp1) edge[] node[auto,midway]{$c_{b,p_1}$} (p2);
\draw(p2) edge[] node[auto,midway]{$a$} (tap2);
\draw(tap2) edge[] node[auto,midway]{$c_{a,p_2}$} (p3);

\draw(p3) edge[bend left] node[auto,midway]{$a$} (tap3);
\draw(p3) edge[bend angle=30,bend left] node[auto,midway]{$b$} (tbp3);
\draw(tap3) edge[bend angle=30,bend left] node[auto,midway,pos=.3]{$c_{a,p_3}$} (p1);
\draw(tbp3) edge[bend angle=30,bend left] node[auto,midway,pos=.3]{$c_{b,p_3}$} (p1);

\node[state] [below=4cm of q1](trash){$\mathit{trash}$};
\draw(trash) edge node[auto,midway]{$\mathit{done}$} (q1);
\draw(trash) edge[loop,out=-60,in=-120,min distance=1cm,looseness=6] node[auto,midway]{$\Gamma \setminus \{\mathit{done}\}$} (trash);

\draw(q1) edge[loop,out=-110,in=-160,min distance=1cm,looseness=6] node[auto,midway,pos=.4]{$\mathit{done}$} (q1);
\draw(p1) edge[loop,out=-110,in=-160,min distance=1cm,looseness=6] node[auto,midway,pos=.4]{$\mathit{done}$} (p1);

\path[->] (trash) edge[<-] +(.4,1.2);
\path[->] (trash) edge[<-] +(.9,.9);
\path[->] (trash) edge[<-] +(1.2,.4);
\end{tikzpicture}\caption{An example construction from the reduction in the proof of~\Cref{thm:PEW-PSPACE}.
Two input automata are shown above and the construction is shown below.
Some transitions are omitted: for all $x \in Q_i$, $y \in Q_j$ with $i \neq j$, the transitions of the transitional letters $c_{\alpha,x}$ fix the states $y$ and $t_{\alpha,y}$; the other omitted transitions go to $\mathit{trash}$.
We have $I = \{q_1,p_1\}$ and $F = \{q_2,p_1\}$.
The shortest accepted word is $bab$, and its corresponding properly extending word from the proof is: $\mathit{done}\;b\;c_{b,q_1}\;c_{b,p_1}\;a\;c_{a,q_2}\;c_{a,p_2}\;b\;c_{b,q_1}\;c_{b,p_3}$.
}\label{fig:ExamplePEWReduction}
\end{figure}
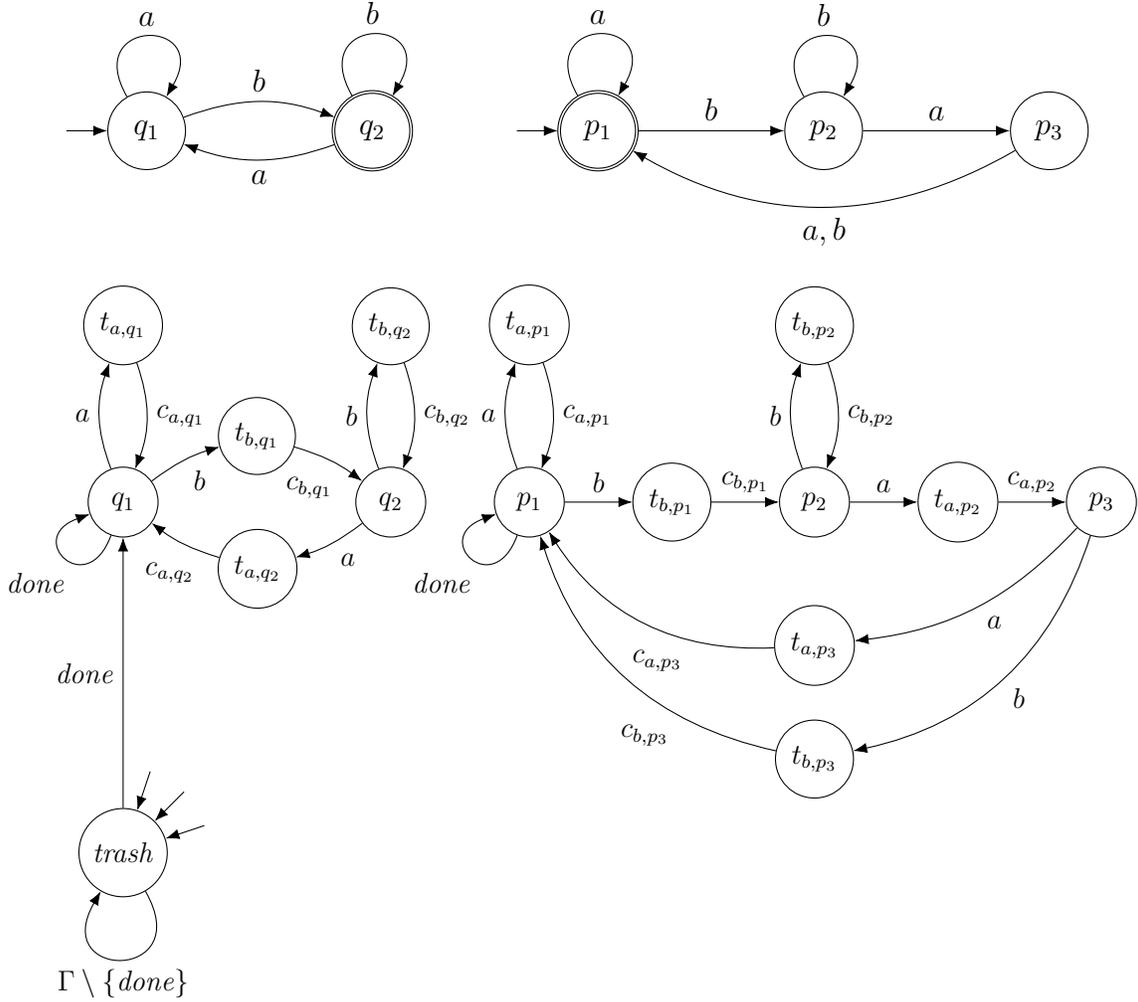

\noindent\textbf{Correctness (accepted word $\Rightarrow$ properly extending word):}

Let $w = \alpha_1 \alpha_2 \ldots \alpha_\ell$, where each $\alpha_i$ is a letter from $\Sigma$, be an accepted word for an $\ell \ge 0$, i.e., for all $i$, we have $\delta_i(\mathit{init}_i,w) = \mathit{final}_i$.
We need to track the current state while applying this word, so let $\mathit{current}(i,j)$ denote the state in $Q_i$ obtained by applying the action of a prefix of $w$ of length $0 \le j \le \ell$ to the initial state, i.e.,
\[ \mathit{current}(i,j) = \delta_i(\mathit{init}_i,\alpha_1 \alpha_2 \ldots \alpha_j) .\]
For instance, we surely have $\mathit{current}(i,0) = \mathit{init}_i$ and $\mathit{current}(i,\ell) = \mathit{final}_i$.

Now we define $w'$ that will be our properly extending word:
\[ w' = \mathit{done} \; \alpha_1\; u_1 \; \alpha_2\; u_2 \; \ldots \; \alpha_\ell \; u_\ell ,\]
where each $u_j$ is a sequence of $m$ transitional letters, one for each of the automata.
These letters depend on the current state in the way that we choose the transition that is used when the letter at the corresponding position in $w$ is applied.
We define:
\[ u_j = c_{\alpha_j,\mathit{current}(1,j)} \; c_{\alpha_j,\mathit{current}(2,j)} \;\ldots \; c_{\alpha_j,\mathit{current}(m,j)} .\]

It remains to show that $w'$ is a properly extending word for $F$.
For $0 \le j \le \ell-1$, let $w'_j$ be the suffix of $w'$ starting from $\alpha_{j+1}$, i.e., $w'_j = \alpha_{j+1} u_{j+1} \ldots \alpha_\ell u_\ell$, and additionally let $w'_\ell = \varepsilon$.

Our auxiliary claim is that $\delta^{-1}(F,w'_j) = \{\mathit{current}(i,j) \mid 1 \le i \le m\}$ and it is a $w'_j$-predecessor of $F$.
We show this by induction descending on $j$.
The base case $j=\ell$ is trivial since $\delta^{-1}(F,\varepsilon) = F$.
For the induction step for $j<\ell$, observe that for each $i$:
\[ \delta^{-1}(\mathit{current}(i,j+1), u_j) = \left\{t_{\alpha_j,\mathit{current}(i,j)}\right\}, \]
which follows due to the transition of the letter $c_{\alpha_j,\mathit{current}(i,j)}$ that maps $t_{\alpha_j,\mathit{current}(i,j)}$ to $\mathit{current}(i,j+1)$ and because the other letters in $u_j$ fix both $\mathit{current}(i,j+1)$ and $t_{\alpha_j,\mathit{current}(i,j)}$.
Then, we have $\delta^{-1}(t_{\alpha_j,\mathit{current}(i,j)},\alpha_j) = \left\{\mathit{current}(i,j)\right\}$ by the definition of the transition of $\alpha_j$, which completes the induction step.
Concluding for $j=0$, we have $\delta^{-1}(F,w'_0) = \{\mathit{current}(i,0) \mid 1 \le i \le m\} = I$.

Finally, $\delta^{-1}(I,\mathit{done}) = I \cup \{\mathit{trash}\}$.
Thus, $I \cup \{\mathit{trash}\}$ is a $w'$-predecessor of $F$ and has size $m+1$.

\noindent\textbf{Correctness (properly extending word $\Rightarrow$ accepted word):}

Let $w'$ be a properly extending word for $F$.
Observe that $w'$ must contain $\mathit{done}$, because otherwise, the preimage could not be larger than $m$, as on all the other letters, for each state except for $\mathit{trash}$, there is at most one incoming transition, and the only incoming transition from $\mathit{trash}$ to another state is on $\mathit{done}$.

Let $w''$ be the suffix of $w'$ after the last occurrence of $\mathit{done}$.
Then $\delta^{-1}(F,w'') = I$, because it must have size $m$ as if it was smaller then would not be a predecessor of $F$, and it can contain only states from $I$, as there are no incoming transitions on $\mathit{done}$ to any other main or transitional states, and it cannot contain $\mathit{trash}$, as $w''$ does not contain $\mathit{done}$.

So $w''$ contains main and transitional letters.
Let split $w'' = u_0\,\alpha_1\,u_1\,\ldots\,\alpha_\ell\,u_\ell$, where each $u_j$ is a word consisting of only transitional letters, each $\alpha_j$ is a main letter, and $\ell \ge 0$.
From $\delta^{-1}(F,w'') = I$ and that $I$ is a $w''$-predecessor of $F$, we have $\delta(I,w'')=F$.
Since there are no transitions on these letters between sets of states $Q_i \cup T_i$ for different $i$s, we also have $\delta(\mathit{init}_i,w'') = \mathit{final}_i$ for each $i$.

We define $w = \alpha_1 \ldots \alpha_\ell$ that will be our accepted word.
It remains to show that for each $i$, we have $\delta_i(\mathit{init}_i,w) = \mathit{final}_i$.
Let $i$ be fixed.
For an $0 \le j \le \ell$, let $w_j$ be the prefix of $w$ of length $j$, i.e., $w_j = \alpha_1 \ldots \alpha_j$, and let $w''_j$ be the prefix of $w''$ that ends with $u_j$, i.e., $w''_j = u_0 \alpha_1 u_1 \ldots \alpha_j u_j$.

We show by induction ascending on $j$ that $\delta(\mathit{init}_i,w''_j) = \delta_i(\mathit{init}_i,w_j)$.

For the base case $j=0$, we have $w_0 = \varepsilon$ and $w''_0 = u_0$.
Recall that all transitional letters act on a main state by either fixing it or mapping to $\mathit{trash}$.
The latter is not possible because we know that $\delta(\mathit{init}_i,w'')=\mathit{final}_i$ and the action of all letters $w''$ fixes $\mathit{trash}$, thus once the current state falls to $\mathit{trash}$, it will stay there.
Hence, $u_0$ fixes $\mathit{init}_i$, so we have $\delta(\mathit{init}_i,w''_0) = \delta(\mathit{init}_i,u_0) = \mathit{init}_i$.

For the inductive step for $j > 1$, we consider $\delta(\mathit{init}_i,w''_j)$.
From the inductive assumption, let $x = \delta(\mathit{init}_i,w''_{j-1}) = \delta_i(\mathit{init}_i,w_{j-1})$, and we consider the next state $\delta(\mathit{init}_i,w''_j) = \delta(x,\alpha_j u_j)$.
So we have $\delta(x,\alpha_j) = t_{\alpha_j,x}$ by the definition of the transitions of $\alpha_j$.
We know that $\delta(t_{\alpha_j,x},u_j)$ is a main state since $u_j$ is either followed by the main letter $\alpha_{j+1}$, whose action maps all transitional states to $\mathit{trash}$, or it is $\mathit{final}_i$ when $j=\ell$.

There is only one transitional letter whose action does not fix $t_{\alpha_j,x}$ nor maps it to $\mathit{trash}$ -- it is $c_{\alpha_j,x}$, thus, $c_{\alpha_j,x}$ must be present in $u_j$.
We have $\delta(t_{\alpha_j,x},c_{\alpha_j,x}) = \delta_i(x,\alpha_j)$ by the definition.
Since the action of transitional letters, which can be in $u_j$, only fix main states or map them to $\mathit{trash}$, where the latter case cannot hold, we know that $\delta(t_{\alpha_j,x},u_j) = \delta(t_{\alpha_j,x},c_{\alpha_j,x}) = \delta_i(x,\alpha_j)$, which is $\delta_i(\delta_i(\mathit{init}_i,w_{j-1}),\alpha_j) = \delta_i(\mathit{init}_i,w_j)$.
Thus, $\delta(\mathit{init}_i,w''_j) = \delta_i(\mathit{init}_i,w_j)$, which completes the inductive step.

Finally, we conclude that $\delta_i(\mathit{init}_i,w) = \delta(\mathit{init}_i,w'') = \mathit{final}_i$, and since this holds for all $i$, the word $w$ is accepted.

\end{proof}
%%%%%%%%%%%%%%%%%%%%%%%%%%%%%%%%%%%%%%%%%%%%%%%%%%%%%%%%%%%%
%%%%%%%%%%%%%%%%%%%%%%%%%%%%%%%%%%%%%%%%%%%%%%%%%%%%%%%%%%%%
%%%%%%%%%%%%%%%%%%%%%%%%%%%%%%%%%%%%%%%%%%%%%%%%%%%%%%%%%%%%
\end{document}